\documentclass[a4paper,UKenglish,cleveref, thm-restate, authorcolumns]{lipics-v2019}
\usepackage[utf8]{inputenc}
\usepackage[T1]{fontenc}
\usepackage{amsmath, amsfonts, amssymb, amsthm,mathtools}
\usepackage{graphicx}
\usepackage{algorithm, algorithmic, calc}
\usepackage{xspace}
\usepackage{tikz}
\usepackage{pgfgantt}
\usetikzlibrary {positioning, shapes, backgrounds,arrows,calc}
\usepackage{bm}
\usepackage{apxproof}
\usepackage{clipboard}
\usepackage[colorinlistoftodos]{todonotes}
\newclipboard{myclipboard}
\usepackage{booktabs}
\usepackage{fontawesome}
\usepackage{pifont}
\usepackage{hyperref}
\usepackage{multirow}
\usepackage{cleveref}
\usepackage{comment}
\usepackage{url}

\newcommand{\tpikiesalgorithm}[1]{\textnormal{Algorithm #1}}

\newcommand{\Pclass}{\textnormal{\sffamily P}\xspace}
\newcommand{\NPclass}{\textnormal{\sffamily NP}\xspace}
\newcommand{\NPHclass}{\textnormal{\sffamily NP-hard}\xspace}
\newcommand{\NPOClass}{\textnormal{\sffamily NPO}\xspace}
\newcommand{\NPC}{\textnormal{\sffamily NP-complete}\xspace}

\newcommand{\APXH}{\textnormal{\sffamily APX-hard}\xspace}

\newcommand{\ThreeSatStar}{\textnormal{\sffamily 3SAT*}\xspace}

\newcommand{\ThreeSat}{\textnormal{\sffamily 3SAT}\xspace}
\newcommand{\conflictgraph}[1][]
{%
	\ifthenelse{\equal{#1}{}}{\mathrm{cliques}}{#1\; \mathrm{cliques}}%
}
\newcommand{\bagmachinerestriction}{M(k)}
\newcommand{\jobmachinerestriction}{M(j)}
\newcommand{\bagmachinespeed}{(p_{k}^{i})_{k\in[b],i\in M}}

\newcommand{\MAXThreeSATSix}{\textnormal{\sffamily MAX 3SAT-6}\xspace}

\newcommand{\nfold}{$n$-fold\xspace}

\newcommand{\confs}{\mathcal{C}}

\DeclarePairedDelimiter\floor{\lfloor}{\rfloor}
\DeclarePairedDelimiter\ceil{\lceil}{\rceil}

\DeclarePairedDelimiter\set{\lbrace}{\rbrace}
\DeclarePairedDelimiterX\sett[2]{\lbrace}{\rbrace}{ #1 \,\delimsize| \,\mathopen{} #2 }

\title{Total Completion Time Minimization for Scheduling with Incompatibility Cliques}
\titlerunning{}

\author{Klaus Jansen}     {Department of Computer Science, Faculty Of Engineering, Kiel University, Germany}{kj@informatik.uni-kiel.de}{}{Supported by DFG Projects "Strukturaussagen und deren Anwendung in Scheduling- und Packungsprobleme", JA 612/20-1 and, JA 612/20-2}

\author{Alexandra Lassota}{Department of Computer Science, Faculty Of Engineering, Kiel University, Germany}{ala@informatik.uni-kiel.de}{}{Supported by DFG Projects "Strukturaussagen und deren Anwendung in Scheduling- und Packungsprobleme", JA 612/20-1 and, JA 612/20-2}

\author{Marten Maack}     {Department of Computer Science, Faculty Of Engineering, Kiel University, Germany}{mmaa@informatik.uni-kiel.de}{}{}

\author{Tytus Pikies}     {Department of Algorithms and System Modelling, ETI Faculty,  Gda\'nsk University of Technology, Poland}{tytpikie@pg.edu.pl}{}{This work was supported by Gda\'nsk University of Technology, grant no. POWR.03.02.00-IP.08-00-DOK/16.}

\authorrunning{K. Jansen, A. Lassota, M. Maack, T. Pikies} 

\Copyright{Authors et al} 



\category{} 

\relatedversion{} 

\supplement{}

\nolinenumbers 

\begin{document}
	
	\maketitle
	
	\nocite{DBLP:journals/siamdm/GoemansW00}
	\nocite{IntroductionToAlgorithms}
	\nocite{}
	\begin{abstract}
		This paper considers parallel machine scheduling with incompatibilities between jobs.
		The jobs form a graph and no two jobs connected by an edge are allowed to be assigned to the same machine.
		In particular, we study the case where the graph is a collection of disjoint cliques. 
		Scheduling with incompatibilities between jobs represents a well-established line of research in scheduling theory and the case of disjoint cliques has received increasing attention in recent years.
		While the research up to this point has been focused on the makespan objective, we broaden the scope and study the classical total completion time criterion.
		In the setting without incompatibilities, this objective is well known to admit polynomial time algorithms even for unrelated machines via matching techniques.
		We show that the introduction of incompatibility cliques results in a richer, more interesting picture.
		Scheduling on identical machines remains solvable in polynomial time, while scheduling on unrelated machines becomes \APXH.
		We identify several more subcases of theses problems that are polynomial time solvable or \NPHclass, respectively.
		Furthermore, we study the problem under the paradigm of fixed-parameter tractable algorithms (FPT).
		In particular, we consider a problem variant with assignment restrictions for the cliques rather than the jobs.
		We prove that it is \NPHclass and can be solved in FPT time with respect to the number of cliques.
		Moreover, we show that the problem on unrelated machines can be solved in FPT time for reasonable parameters, e.g., the parameter pair: number of machines and maximum processing time.
		The latter result is a natural extension of known results for the case without incompatibilities and can even be extended to the case of total weighted completion time.
		All of the FPT results make use of $n$-fold Integer Programs that recently have received great attention by proving their usefulness for scheduling problems.
	\end{abstract}

	\section{Introduction}
	\label{sec:introduction}
	Consider a task system under difficult conditions like high electromagnetic radiation, or with an unstable power supply. 
	Due to the environmental conditions, users prepare tasks in groups and want the jobs in a given group to be scheduled on different processors. That assures that even if a few processors fail, another processor will be able to execute at least part of the jobs.
	Due to the instability, our system even might stop working completely and in this case all jobs that are done only partially have to be scheduled again.
	As observed  in~\cite{book:TheoryOfScheduling} and further pointed out in~\cite{DBLP:journals/cacm/BrunoCS74}, the sum of completion times criterion tends to reduce the mean number of unfinished jobs at each moment in the schedule.
	For this reason we would like to minimize the sum of completion times of the jobs respecting the additional reliability requirement given by the groups. 
	In the following, we discuss the problems motivated by this scenario more formally.

	\subparagraph*{Problem.}

	In the classical problem of scheduling on parallel machines, a set $J$ of $n$ jobs, a set $M$ of $m$ machines, and a processing time function $p$ are given.
	The processing times are of the form $p: J \rightarrow \mathbb{N}$ if the machines are identical or of the form $p: J \times M \rightarrow \mathbb{N}\cup\set{\infty}$ if the machines are unrelated. 
	That is, the processing time of a job does or does not, respectively, depend on the machine to which the job is assigned to.
	For brevity, we usually write $p_{j}$ or $p_{j}^{i}$ instead of $p(j)$ or $p(j,i)$ for each job $j$ and machine $i$.
	The goal is to find a schedule of the jobs on the machines, which minimizes a given objective function.
	A schedule in this setting is an assignment from jobs to machines and starting times.
	However, by the fact that for any machine we can order the jobs assigned to it optimally, according to Smith's rule~\cite{SmithRule}, for brevity we do not specify the starting times explicitly. 
	The completion time $C_j$ of $j$ is given by the sum of its starting and processing times.
	Probably the most studied objective function is the minimization of the makespan $C_{\max} = \max_j C_j$, directly followed by the minimization of the total completion time objective $\sum C_j$ or the sum of weighted completion times~$\sum w_jC_j$.
	In this paper, we use the three-field notation prevalent in scheduling theory.
	For instance, makespan minimization on identical machines is abbreviated as $P||C_{\max}$ and minimization of the total completion time on unrelated machines as $R||\sum C_j$.
	For a general overview of scheduling notation we refer the reader to~\cite{DBLP:books/daglib/0010885}.
	
	All of the scheduling problems discussed so far are fundamental and often studied with respect to additional constraints. 
	One line of research considers incompatibilities between jobs in the sense that some jobs may not be processed by the same machine.
	More formally, an incompatibility graph $G=(J,E)$ is part of the input, and an edge $\set{j,j'}\in E$ signifies that in a feasible schedule $j$ and $j'$ cannot be assigned to the same machine.
	In this paper, we study variants of $P||\sum (w_j) C_j$ and $R||\sum (w_j) C_j$ in which the incompatibility graph is a collection of cliques corresponding to the groups of jobs mentioned above. 
	In the three-field notation, we denote the class to which the incompatibility graph belongs in the middle, e.g. $P|\conflictgraph{}|\sum (w_j) C_j$ or $R|\conflictgraph{}|\sum (w_j) C_j$.

	\subparagraph*{Related Work.}

	First note that both $P||\sum w_jC_j$ and $P||C_{\max}$ are well-known to be strongly NP-hard.
	On the other hand, $P||\sum C_j$ can be solved in polynomial time via a simple greedy method (see \cite{book:TheoryOfScheduling}) and even $R||\sum C_j$ can be shown to be in $\Pclass$ via matching techniques \cite{DBLP:journals/cacm/BrunoCS74}.
	
	Scheduling with incompatibilities has first been considered in the 1990's by Jansen, Bodlaender and Woeginger~\cite{DBLP:journals/dam/BodlaenderJW94} who studied $P||C_{\max}$ with incompatibilities between jobs in the sense used in this paper.
	Among other things they presented an approximation algorithm which approximation ratio depends on the quality of a coloring for the incompatibility graph.
	The result yields constant approximation algorithms for subproblem where incompatibility graph can be colored in polynomial time with with constant number of colors, which is less than the number of the machines.
	Furthermore, Jansen and Bodlaender~\cite{DBLP:conf/mfcs/BodlaenderJ93} presented hardness results in the same setting for cographs, bipartite graphs and interval graphs.
	More recently, there has been a series of results for the context with uniformly related machines and unit processing times \cite{FurmanczykK17,DBLP:journals/dam/FurmanczykK18,DBLP:journals/cor/MallekBB19} for several classes of incompatibility graphs like (complete) bipartite graphs, forests, or $k$-chromatic cubic graphs.
	In 2012, Dokka, Kouvela, and Spieksma \cite{DBLP:journals/orl/DokkaKS12} presented approximation and inapproximability results for the so called multi-level bottleneck assignment problem.
	This problem can be seen as a variant of $P|\conflictgraph{}|C_{\max}$ in which each clique has the same size and each machine has to receive exactly one job from each clique.
	However, the exact setting studied in the present paper (with respect to incompatibilities) was introduced only recently by Das and Wiese~\cite{DBLP:conf/esa/DasW17} who called the cliques \emph{bags}.
	They obtained a PTAS for $P|\conflictgraph{}|C_{\max}$ and showed that (unless $\Pclass=\NPclass$) there is no constant approximation algorithm for the restricted assignment variant $P|\conflictgraph{}, \jobmachinerestriction|C_{\max}$, i.e., the case in which each job $j$ may only be processed on a given set $M(j)$ of machines eligible for $j$.
	Moreover, they gave an $8$-approximation for the special case $P|\conflictgraph{}, \bagmachinerestriction|C_{\max}$ in which jobs in the same clique have the same restrictions, i.e., sets $M(k)$ of eligible machines are given for each clique $k\in[b]$.
	This line of research was continued by two groups.
	In particular, Grage, Jansen and Klein~\cite{DBLP:conf/spaa/GrageJK19} obtained an EPTAS for $P|\conflictgraph|C_{\max}$, and Page and Solis-Oba~\cite{DBLP:journals/tcs/PageS20} considered a variant of $R|\conflictgraph|C_{\max}$ where the number of machine types and cliques is restricted and obtained a PTAS among many other results.
	Two machines have the same type if the processing time of each job is the same on both of them.

	Finally, we also consider fixed-parameter tractable (FPT) algorithms for scheduling problems.
	A good overview on this line of research is provided in a survey by Mnich and van Bevern \cite{DBLP:journals/cor/MnichB18}.
	The most notable result in our context is probably a work due to Knop and Koutecký \cite{DBLP:journals/scheduling/KnopK18} who used so-called $n$-fold Integer Programs to prove (among other things) two FPT results for $R||\sum w_jC_j$.
	In particular, $R||\sum w_jC_j$ is FPT with respect to the number of machines and the number of different job kinds $\vartheta$ , and also FPT with respect to the maximum processing time, the number of different job kinds $\vartheta$, and the number  of distinct machine types $\kappa$. These results were generalized and greatly extended by Knop et al.\ in \cite{DBLP:journals/corr/abs-1909-07326}. In their work, they introduce a general framework for solving various configuration ILPs by modeling them as (an extended version of) the Monoid Decomposition problem. This allows to solve many problems with different kinds of objects (for example, jobs with release times and due dates) and locations (for example, unrelated machines) and (linear or non-linear) objectives in FPT time with plenty different, natural parameterizations.

	\subparagraph*{Results and Methodology.}

	The results of our paper can be divided into three groups.
	The first one is comprised of polynomial time algorithms for several variants of $P|\conflictgraph{}|\sum C_j$ and $R|\conflictgraph{}|\sum C_j$.
	These results are based on classical approaches like flow and matching techniques, dynamic programming, greedy algorithms, and exhaustive search.
	They are presented in Section \ref{sec:polyalgos}.
	Next, we present hardness results in Section \ref{sec:hardness}.
	In the reductions some ideas previously used for variants of $P||C_{\max}$ and $R||C_{\max}$ (see, e.g., \cite{DBLP:journals/algorithmica/EbenlendrKS14,DBLP:journals/jcss/ChenJZ18,DBLP:conf/stacs/0011MYZ17,DBLP:conf/stacs/MaackJ20}) are reused.
	Finally, we present several FPT results all of which are based on $n$-fold Integer Programs which have proven increasingly useful in the context of scheduling in recent years, see, e.g., \cite{DBLP:journals/scheduling/KnopK18,DBLP:conf/esa/KnopKM17,DBLP:conf/innovations/JansenKMR19}.
	These results are discussed in Section \ref{sec:fptalgos}.
	All of our results are summarized in Table \ref{table:results}.
	\begin{table}
		\caption{An overview of the results of this paper.
			For the classical polynomial time algorithms the running times are listed.}
		\centering
		\begin{tabular}{ll}
			\toprule
			Problem & Result \\ \midrule
			$P|\conflictgraph{}|\sum C_j$ & $O(m^2n + mn^{5/2})$\\
			$R|\conflictgraph{}, p_{j}^{i} \in \{1, \infty \}|\sum C_j$ & $O(m^2n^3\log mn)$\\
			$R|\conflictgraph[2], p_{j}^{i} \in \{p_1 \le p_2\}|\sum C_j$ & $O(m^2n^4\log mn)$ \\
			$P|\conflictgraph[b], \bagmachinerestriction|\sum C_j$ & $O(m^{O(b^{(b+1)})}n^{3}\log mn)$\\
			$Rm |\conflictgraph{}, p_{j}^{i} \in \{a_1, \ldots, a_k\}| \sum C_j$ &  $O(n^{2km}nm^m)$\\
			$R|\conflictgraph{}, \jobmachinerestriction, \bagmachinespeed|\sum C_j$ &  $O(m^2n^4\log mn))$\\
			$P|\conflictgraph{}, \bagmachinerestriction, p_i \in \{p_1 < p_2 < 2p_1\}|\sum C_j$ &\APXH\\
			$R|\conflictgraph[2], p_{j}^{i} \in \{p_1 < p_2 < p_3\}|\sum C_j$ &\NPC\\
			$P|\conflictgraph{}, p_j \in \{p_1 < p_2\},\bagmachinerestriction|\sum C_j$  &\NPC\\
			$R|\conflictgraph{}, p_j \in \{p_1 < p_2\}|\sum C_j$  &\NPC\\
			$P|\conflictgraph{}, \bagmachinerestriction|\sum C_j$  & FPT w.r.t. $b$\\
			$R|\conflictgraph{}|\sum w_j C_j$ & FPT w.r.t. $m$, $p_{\max}$, $\vartheta$\\
			$R|\conflictgraph{}|\sum w_j C_j$ & FPT w.r.t. $\kappa$, $b$, $p_{\max}$, $\vartheta$\\\bottomrule		\end{tabular}
		\label{table:results}
	\end{table}
	
	We briefly discuss the results, establish links between them, and introduce the missing notation.
	First of all, we show that the problem $P|\conflictgraph{}|\sum C_j$ remains in $\Pclass$ while $R|\conflictgraph{}|\sum C_j$ -- unlike $R||\sum C_j$ -- is $\NPC$. 
	Hence, the introduction of incompatibility cliques results in a richer, more interesting picture which we explore in more detail.
	In particular, the problem remains $\NPC$ even in the case with only two cliques and three distinct processing times $R|\conflictgraph[2], p_{j}^{i} \in \{p_1 < p_2 < p_3\}|\sum C_j$, and in the case with only two distinct processing times and arbitrarily many cliques $R|\conflictgraph{}, p_j \in \{p_1 < p_2\}|\sum C_j$.
	On the other hand, the case with two cliques and two processing times $R|\conflictgraph[2], p_{j}^{i} \in \{p_1 \le p_2\}|\sum C_j$, or the case with many cliques and two processing times~$1$ and $\infty$, denoted as $R|\conflictgraph{}, p_{j}^{i} \in \{1, \infty \}|\sum C_j$, are both in $\Pclass$. 
	Furthermore, a setting derived from our motivational example turns out to be polynomial time solvable, that is, the variant of $R|\conflictgraph{}|\sum C_j$ in which jobs belonging to the same clique have the same processing times, and hence can be seen as copies of the same job.
	This remains true even if we introduce additional \emph{job dependent} assignment restrictions.
	The corresponding problem is denoted as $R|\conflictgraph{}, \jobmachinerestriction, \bagmachinespeed|\sum C_j$.
	Note that this setting is closely related to the case with clique dependent assignment restrictions introduced by Das and Wiese \cite{DBLP:conf/esa/DasW17}.
	We study this case as well and prove it to be \NPC and even \APXH already for the case with only two processing times $P|\conflictgraph{}, p_j \in \{p_1 < p_2\},\bagmachinerestriction|\sum C_j$. 
	On the other hand, it can be solved in polynomial time if the number of cliques is constant even if there are arbitrarily many processing times $P|\conflictgraph[b], \bagmachinerestriction|\sum C_j$.
	While the last result relies on costly guessing steps, we can refine it using $n$-fold Integer Programs yielding an FPT time algorithm with respect to $b$ for $P|\conflictgraph{},\bagmachinerestriction|\sum C_j$.
	Furthermore, we revisit FPT results due to Knop and Koutecký \cite{DBLP:journals/scheduling/KnopK18} for $R||\sum w_j C_j$.
	Careful extensions of the ILPs considered in this work yield that $R|\conflictgraph{}|\sum w_j C_j$ is FPT with respect to $m$, $p_{\max}$, and $\vartheta$.
	In particular $b$ is not needed as a parameter in this setting.
	If we remove $m$ from parameters, we get FPT running times with respect to $b$, $p_{\max}$, and $\vartheta$.
	Interestingly, the setting with a constant number of machines and processing times $Rm |\conflictgraph{}, p_{j}^{i} \in \{a_1, \ldots, a_k\}| \sum C_j$ is in \Pclass.
	Hence, it would be interesting if FPT results with respect to the number of distinct processing times are achievable in this setting.
	For a discussion of further open problems we refer to Section \ref{sec:open_problems}.
	
	\section{Polynomial Time Algorithms}\label{sec:polyalgos}
	In this chapter, the polynomial time algorithms are presented.
	For all of the problems we construct flow networks, sometimes disguised as matchings in bipartite graphs.
	We also use a greedy approach for $P||\sum C_j$, dynamic programming, as well as exhaustive search.
	
	In the following, we understand by $D = (V, A, capacity, cost)$ a digraph with the set of vertices $V$, the set of arcs $A$, capacities on the arcs given by a function $capacity: E \rightarrow N$, and the cost of a flow by an arc given by $cost: E \rightarrow N$. 
	A directed edge between $v_1 \in V$ and $v_2 \in V$ is denoted by $(v_1, v_2)$.
	
	\subsection{A Polynomial Time Algorithm for Identical Machines}
	Let us begin with a key procedure for the algorithm.
	In a nutshell, we prove that a perfect matching in the vertices of the graph constructed in line \ref{algline:IdentMachinesPerfectMatch} of \cref{alg:incompsolver} corresponds to a reassignment of the jobs in $S$ in a way that the assignment of the jobs to $m_1, \ldots, m_{i-1}$ is not changed and that $m_i$ is given a set of compatible jobs.
	Without loss of generality we assume that each clique $V_i$ consists of exactly $m$ jobs; if this is not the case we add dummy jobs with processing time $0$.
	Notice also that in any schedule the jobs can be divided into layers.
	Precisely, the layers are formed of the jobs that are scheduled as last on their machines, then the ones scheduled before the last ones, \ldots, and as first (which correspond to $b$-th layer). 
	We can exchange the jobs that are on a given layer without increasing the total completion time, because the job scheduled as last on a machine contributes once to the total completion time, the jobs scheduled exactly before last twice, etc.
	\begin{algorithm}
		\begin{algorithmic}[1]
			\REQUIRE A set of cliques $V_1 \cup \ldots \cup V_b$, number $1 \le i \le m$, a schedule~$S$ such that machines $m_1, \ldots, m_{i-1}$ have compatible jobs assigned.
			\ENSURE A schedule with the total completion time equal to the total completion time of~$S$, where jobs on $m_1, \ldots, m_{i}$ are independent. 
			\STATE $V_L = \{v_{L}[1], \ldots, v_{L}[b]\}$.
			\STATE $V_B = \{v_B[1], \ldots, v_B[b]\}$. \label{algline:VB}
			\STATE Construct $E$ by connecting $v_B[i]$ to the vertex $v_L[j]$ iff on machines $m_i, \ldots, m_m$ there is a job from $V_i$ scheduled as $j$-th.
			\STATE Let $M'$ be a perfect matching in $(V_L \cup V_B, E)$. \label{algline:IdentMachinesPerfectMatch}
			\FOR {$l = 1, \ldots, b$}
			\STATE Let $\{v_L[l], v_B[j]\} \in M'$.
			\STATE Exchange the job on position $l$-th on $m_i$ with any job from $V_j$ assigned to position $l$-th on $m_i, \ldots, m_m$
			\ENDFOR
			\RETURN $S$
		\end{algorithmic}
		\caption{IncompatibilitySolving$(i, S)$}
		\label{alg:incompsolver}
	\end{algorithm}
	
	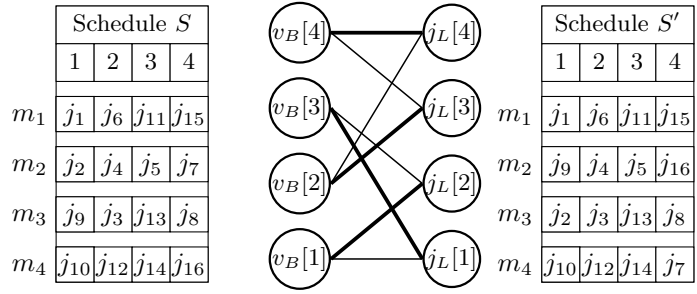
\begin{figure}
		\centering
		\small
		\begin{ganttchart}
			[
			y unit title=0.5cm,
			y unit chart=0.665cm,
			title label anchor/.style={below=-1.6ex},
			title height=1,
			progress label text={},
			bar height=0.7,
			]
			{1}{4}
			\gantttitle{Schedule $S$}{4} \\
			\gantttitlelist{1,...,4}{1} \\
			\ganttbar{$m_1$}{1}{1} \ganttbar[inline]{$j_{1}$}{1}{1} \ganttbar[inline]{$j_{6}$}{2}{2} \ganttbar[inline]{$j_{11}$}{3}{3} \ganttbar[inline]{$j_{15}$}{4}{4}\\
			\ganttbar{$m_2$}{1}{1} \ganttbar[inline]{$j_{2}$}{1}{1} \ganttbar[inline]{$j_{4}$}{2}{2} \ganttbar[inline]{$j_{5}$}{3}{3} \ganttbar[inline]{$j_{7}$}{4}{4}\\
			\ganttbar{$m_3$}{1}{1} \ganttbar[inline]{$j_{9}$}{1}{1} \ganttbar[inline]{$j_{3}$}{2}{2} \ganttbar[inline]{$j_{13}$}{3}{3} \ganttbar[inline]{$j_{8}$}{4}{4}\\
			\ganttbar{$m_4$}{1}{1} \ganttbar[inline]{$j_{10}$}{1}{1} \ganttbar[inline]{$j_{12}$}{2}{2} \ganttbar[inline]{$j_{14}$}{3}{3} \ganttbar[inline]{$j_{16}$}{4}{4}
		\end{ganttchart}
		\hspace{0.5cm}
		\begin{tikzpicture}
		[place/.style={circle,draw=black!100,line width=0.3mm,inner sep=0pt,minimum size=6mm},
		subordinate/.style={circle,draw=black!100,line width=0.3mm,inner sep=0pt,minimum size=2mm},
		myline/.style={line width=0.2mm},
		matching/.style={line width=0.5mm}]
		\node at ( 2,0)   (j1) [place] {$j_{L}[1]$};
		\node at ( 2,1)   (j2) [place] {$j_{L}[2]$};
		\node at ( 2,2)   (j3) [place] {$j_{L}[3]$};		
		\node at ( 2,3)   (j4) [place] {$j_{L}[4]$};	
		
		\node at ( 0, 0)   (b1) [place] {$v_B[1]$};
		\node at ( 0, 1)   (b2) [place] {$v_B[2]$};
		\node at ( 0, 2)  (b3) [place] {$v_B[3]$};
		\node at ( 0, 3)   (b4) [place] {$v_B[4]$};
		
		\draw [matching]  (b1.east) -- (j2.west);\draw [myline]  (b1.east) -- (j1.west);
		\draw [matching]  (b2.east) -- (j3.west);\draw [myline]  (b2.east) -- (j4.west);
		\draw [matching]  (b3.east) -- (j1.west);\draw [myline]  (b3.east) -- (j2.west);
		\draw [myline]  (b4.east) -- (j3.west);\draw [matching]  (b4.east) -- (j4.west);
		
		\end{tikzpicture}
		\begin{ganttchart}
			[
			y unit title=0.5cm,
			y unit chart=0.665cm,
			title label anchor/.style={below=-1.6ex},
			title height=1,
			progress label text={},
			bar height=0.7,
			]
			{1}{4}
			\gantttitle{Schedule $S'$}{4} \\
			\gantttitlelist{1,...,4}{1} \\
			\ganttbar{$m_1$}{1}{1} \ganttbar[inline]{$j_{1}$}{1}{1} \ganttbar[inline]{$j_{6}$}{2}{2} \ganttbar[inline]{$j_{11}$}{3}{3} \ganttbar[inline]{$j_{15}$}{4}{4}\\
			\ganttbar{$m_2$}{1}{1} \ganttbar[inline]{$j_{9}$}{1}{1} \ganttbar[inline]{$j_{4}$}{2}{2} \ganttbar[inline]{$j_{5}$}{3}{3} \ganttbar[inline]{$j_{16}$}{4}{4}\\
			\ganttbar{$m_3$}{1}{1} \ganttbar[inline]{$j_{2}$}{1}{1} \ganttbar[inline]{$j_{3}$}{2}{2} \ganttbar[inline]{$j_{13}$}{3}{3} \ganttbar[inline]{$j_{8}$}{4}{4}\\
			\ganttbar{$m_4$}{1}{1} \ganttbar[inline]{$j_{10}$}{1}{1} \ganttbar[inline]{$j_{12}$}{2}{2} \ganttbar[inline]{$j_{14}$}{3}{3} \ganttbar[inline]{$j_{7}$}{4}{4}	\end{ganttchart}
		
		\caption
		{
			An illustration of an application of \cref{alg:incompsolver}.
			Let the set of cliques be given by $\{j_{1}, j_{2}, j_{3}, j_{4}\}, \{j_5, j_6, j_7, j_8\}, \{j_{9}, j_{10}, j_{11}, j_{12}\}, \{j_{13}, j_{14}, j_{15}, j_{16}\}$ and let $i=2$ (which means that $m_1$ has already a set of compatible jobs assigned).
			For clarity we identify the labels of the vertices with the labels of the jobs.
			Notice how using a matching in the constructed graph the jobs can be exchanged in a way that $m_2$ has only compatible jobs assigned.
		}
		\label{figure:IncompatibilitySolvingIllustration}
	\end{figure}
	
	\begin{theorem}[\cite{HallTheorem}]
		\label{theorem:HallTheorem}
		A bipartite graph $(A \cup B, E)$ has a matching that saturates $A$ if and only if $|N(S)| \ge |S|$ for all $S \subseteq A$.
	\end{theorem}
	\begin{lemma}
		\label{lemma:IncompatibilitySolving}
		Let $S$ be a schedule for an instance of $P|\conflictgraph{} = V_1 \cup \ldots \cup V_b|\Sigma C_j$, such that each of the machines $m_1, \ldots, m_{i-1}$ has compatible jobs assigned in~$S$ and $b$~jobs are assigned to each machine.
		\tpikiesalgorithm{\ref{alg:incompsolver}} constructs in $O(bm + b^{5/2})$ time a schedule such that  each of the machines $m_1, \ldots, m_i$ has compatible jobs assigned, each of the machines has $b$ jobs and the total completion time of the new schedule is equal to the total completion time of $S$. 
	\end{lemma}
	\begin{proof}
		Remember that each of the cliques has exactly $m$ jobs.
		We prove that it is always possible to exchange the jobs inside the layers $1, \ldots, b$ for machines $m_i, \ldots, m_m$, so that the load on $m_i$ consists of compatible jobs. 
		Consider the structure of the graph constructed in \tpikiesalgorithm{\ref{alg:incompsolver}}.
		Take any $V_B' \subseteq V_B$. 
		Notice that the cliques corresponding to this vertices, have exactly $m-i+1$ jobs on the machines $m_i, \ldots, m_{m}$.
		A layer $i$ on the machines $m_{i}, \ldots, m_m$ has exactly $m-i+1$ jobs in total. 
		Hence, clearly the size of neighbors of $V_B'$ in $V_L$ is at least $|V_B'|$; hence by \cref{theorem:HallTheorem} there is a perfect matching in the graph.
	\end{proof}
	
	Consider an instance of $P|\conflictgraph{} |\Sigma C_j$.
	Assume that each of the cliques have $m$ jobs.
	If this is not the case, add jobs with processing time $0$.
	Order the jobs non-increasingly according to their processing times and schedule them in a round robin fashion without respect to the cliques, which is by the Smith's Rule optimal \cite{SmithRule}.
	By \cref{lemma:IncompatibilitySolving} we may easily construct a method to change the schedule to respect the cliques.
	Hence, the following theorem follows:
	\begin{theorem} 
		\label{theorem:IdenticalMachinesTotalCompletion}
		\Copy{theorem:IdenticalMachinesTotalCompletion:text}
		{
			$P|\conflictgraph{} |\Sigma C_j$ can be solved to optimality in $O(m^2n + mn^{5/2})$ time. 
		}
	\end{theorem}
	\begin{proof}
		First let us notice that by adding the jobs with processing time equal to $0$ we do not increase the total completion time, because we can always schedule these jobs before the "normal" jobs, even when scheduling with respect to cliques.
		Then we may use a round robin to obtain an optimal schedule without respect to the cliques.
		The round robin means an assignment of the job with largest processing time to the position $(1,1)$, \ldots, the job with $m$-th largest processing time to $(m,1)$, the job with $m+1$-th largest processing time to $(1,2)$, etc.
		Hence by \cref{lemma:IncompatibilitySolving} and a simple inductive argument the correctness follows. 
		The complexity follows from $b = O(n)$, and from Hopcrof-Karp algorithm.
	\end{proof}	
	
	
	\subsection{Exact Polynomial Algorithms For Unrelated Machines}
	
	\begin{theorem}[\cite{NetworkFlows}, \cite{IntroductionToAlgorithms}]
		\label{theorem:MaximumFlowMinimumCost}
		For a network given by a digraph $(V, A, capacity, cost)$ the maximum flow with minimum cost can be computed in $O(|V|U(|A| + |V|\log|V|))$ time, where $U$ is the maximum flow value, using the Successive Shortest Path Algorithm~\cite{NetworkFlows} and Dijkstra's Algorithm with Fibonacci heaps~\cite{IntroductionToAlgorithms}.
	\end{theorem}
	
	We can solve $R|\conflictgraph{}, p_{j}^{i} \in \{1, \infty \}|\sum C_j$ by constructing a suitable flow network.
	Assume that there exists a schedule with finite cost; otherwise use an algorithm following from \cref{theorem:IdenticalMachinesTotalCompletion}.
	In this case each of the machines can do at most one job from a clique.
	The total completion time of the jobs assigned to a machine is a function of the number of such jobs.
	We refer the reader to \cref{figure:AlgorithmUnitTimeBags} for an overview of a sample construction.
	\begin{theorem}	
		\label{theorem:SchedulingByFlows}
		\Copy{theorem:SchedulingByFlows:text}
		{
			The problem $R|\conflictgraph{}, p_{j}^{i} \in \{1, \infty \}|\sum C_j$ can be solved in $O(m^2n^3\log mn)$ time.
		}
	\end{theorem}
	\begin{figure}
		\centering
		\small
		\begin{tikzpicture}
		[place/.style={circle,draw=black!100,line width=0.3mm,inner sep=1pt,minimum size=8mm},
		subordinate/.style={circle,draw=black!100,line width=0.3mm,inner sep=1pt,minimum size=7mm},
		myline/.style={line width=0.3mm}]
		\node at ( 0,2)   (s) [subordinate] {$s$};
		\node at ( 2,2)   (sprim) [subordinate] {$s'$};
		
		\node at ( 4.25,0)  (m1j1) [place] {\small $m_1,1$};
		\node at ( 4.25,1.0)    (m1jx) [place] {\small $\ldots$};
		\node at ( 4.25,2)    (m1jn) [place] {\small $m_1,n$};
		\node at ( 4.25,3)   (mxjx) [place] {\small $\ldots$};
		\node at ( 4.25,4) (mmjn) [place] {\small $m_m, n$};
		
		\node at ( 6.5,1.0)  (m1) [place] {$m_1$};
		
		\node at ( 8.5,1)    (m1b1) [place] {\small $m_1,1$};
		\node at ( 8.5,3)   (m1bx) [place] {\small $\ldots$};
		\node at ( 8.5,4)   (m1bb) [place] {\small $m_1,b$};
		
		\node at ( 11, 0)  (j1b1) [place] {\small $v_1^{V_1}$};
		\node at ( 11 ,1)  (jxb1) [place] {\small $\ldots$};
		\node at ( 11,2)  (jnb1) [place] {\small $v^{V_1}_{|V_1|}$};
		\node at ( 11,3)  (jxbx) [place] {\small $\ldots$};
		\node at ( 11,4)  (jnbb) [place] {\small $v^{V_b}_{|V_b|}$};
		
		\node at ( 13.25,2)  (t) [subordinate] { $t$};
		
		\draw [myline,->]  (s) -- node[midway, fill=white]  {$(c,0)$} (sprim);
		
		\draw [myline,->]  (sprim.east)  |- node[near end, fill=white] {$(1,1)$} (m1j1);
		\draw [myline,->]  (sprim.east)  |- node[near end, fill=white] {$(\ldots)$} (m1jx);
		\draw [myline,->]  (sprim.east)  |- node[near end, fill=white] {$(1,n)$}  (m1jn);
		\draw [myline,->]  (sprim.east)  |- node[near end,, fill=white] {$(\ldots)$} (mxjx);
		\draw [myline,->]  (sprim.east)  |- node[near end,, fill=white] {$(1,n)$} (mmjn);
		
		\draw [myline,->]  (m1jx.east)  -- (m1.west);
		\draw [myline,-]  (m1jx.east)  -| node[near start, fill=white] {$(1,0)$} ($(m1jx.east)!.85!(m1.west)$);
		\draw [myline,-]  (m1j1.east)  -| node[near start, fill=white] {$(1,0)$} ($(m1jx.east)!.85!(m1.west)$);
		\draw [myline,-]  (m1jn.east)  -| node[near start, fill=white] {$(1,0)$} ($(m1jx.east)!.85!(m1.west)$);
		
		\draw [myline,->]  (m1.east)  |- node[near end, fill=white] {$(1,0)$} (m1b1);
		\draw [myline,->]  (m1.east)  |- node[near end, fill=white] {$(1,0)$}(m1bx);
		\draw [myline,->]  (m1.east)  |- node[near end, fill=white] {$(1,0)$} (m1bb);
		
		\draw [myline,->]  (m1b1.east)  |- node[near end, fill=white] {\small $(1, cost)$} (j1b1);
		\draw [myline,->]  (m1b1.east)  |- node[near end, fill=white] {\small $(1, cost)$} (jxb1);
		\draw [myline,->]  (m1b1.east)  |- node[near end, fill=white] {\small $(1, cost)$} (jnb1);
		\draw [myline,->]  (m1bx.east)  |- node[near end, fill=white] {\small $(1, cost)$} (jxbx);
		\draw [myline,->]  (m1bb.east)  |- node[near end, fill=white] {\small $(1, cost)$} (jnbb);
		
		\draw [myline,->]  (jnb1.east)  -- (t.west);
		\draw [myline,-]  (jnb1.east)  -| node[near start, fill=white] {$(1,0)$} ($(jnb1.east)!.85!(t.west)$);
		\draw [myline,-]  (j1b1.east)  -| node[near start, fill=white] {$(1,0)$} ($(jnb1.east)!.85!(t.west)$);
		\draw [myline,-]  (jxb1.east)  -| node[near start, fill=white] {$(1,0)$} ($(jnb1.east)!.85!(t.west)$);
		\draw [myline,-]  (jxbx.east)  -| node[near start, fill=white] {$(1,0)$} ($(jnb1.east)!.85!(t.west)$);
		
		\draw [myline,-]  (jnbb.east)  -| node[near start, fill=white] {$(1,0)$} ($(jnb1.east)!.85!(t.west)$);
		\end{tikzpicture}
		\caption{
			An illustration of the flow network constructed for \cref{theorem:SchedulingByFlows}. 
			The first field of an edge's label is its capacity, the second one is the cost per unit of flow.
			For a $v_j \in V_k$ the $cost$ field in an arc $((m_i, k),v_j)$ is $p_j^i-1$, hence it is $0$ or $\infty$.
			Notice how the cost of a flow by the network corresponds to a cost of a schedule and how a capacity of an edge forces the flow to "respect the cliques".
		}
		\label{figure:AlgorithmUnitTimeBags}
	\end{figure}
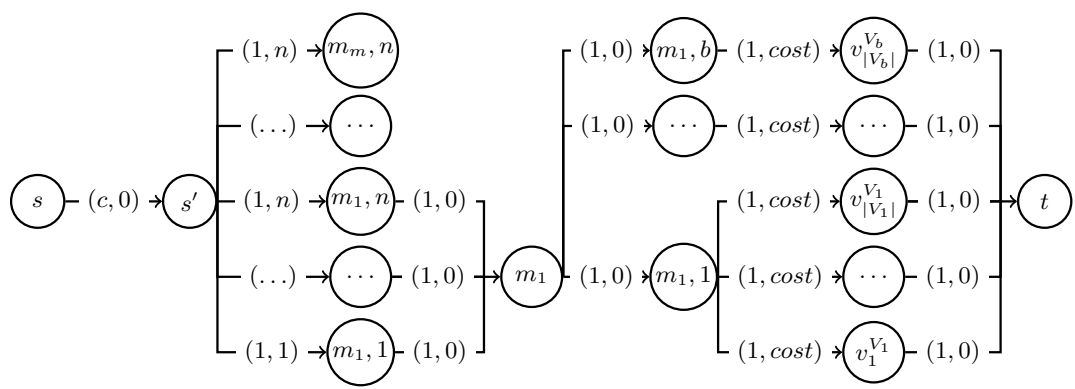
	
	\label{apx:FirsAlgorithmDescription}
	\begin{proof}
		Consider \tpikiesalgorithm{\ref{alg:identical_unit_time_restricted_assignment}}.
		Let us make two observations.
		The arc between $s$ and $s'$ has capacity other than $n$ only because the algorithm is reused for another problem.
		The arcs $A_2$ and $A_3$ could be merged, but for the clarity of the notation they are separated.
		To see that the algorithm works consider the interpretation of the flow in the constructed network.
		Assume that there is a schedule with a finite total completion time.
		In this case we can easily construct a flow with the cost equal to cost of the schedule, by considering to which machines the jobs are assigned and how many jobs are assigned to a given machine.
		Consider an integral flow with the minimum cost; notice that by the fact that all the capacities are integral such a flow exists.
		We show that it corresponds to a schedule with the minimum cost.
		A flow on an arc $(s', (m_i, l))$ corresponds to an assignment of a job as $l$-th on machine $m_i$.
		If $(s', (m_i, l))$ has a flow, all the arcs $(s', (m_i, l-1)), \ldots, (s', (m_i, 1))$ have to have a flow, due to the fact that the cost of the flow is minimal.
		Moreover, the cost of the flow by $(s', (m_i, l))$ is exactly the contribution of the jobs scheduled as $l$-th on the machine.
		Due to the fact that the arcs in $A_4$ have capacity $1$ the cliques restrictions are satisfied.
		Finally notice that the flows by the arcs in $A_5$ correspond to an assignment of the jobs to the machines and we can order them in any way on the machines.
		
		The complexity follows from an observation that the number of vertices and the number of arcs are both $O(nm)$ and that the maximum flow value is $O(n)$.
		Hence by \cref{theorem:MaximumFlowMinimumCost} the problem can be solved in $O(m^2n^3\log mn)$ time.	\end{proof}
	\begin{algorithm}
		\begin{algorithmic}[1]
			\REQUIRE A set of cliques $V_1, \ldots, V_b$, a set of $m$ unrelated machines $M =\{m_1, \ldots, m_m\}$, a parameter $c$.
			\ENSURE An optimal schedule.
			\STATE 	 Construct the following digraph.
			\begin{align*}
			V &\leftarrow \{s, s'\} \cup (M \times \{1, \ldots, n\})  \cup M \cup (M \times \{V_1, \ldots, V_b\}) \cup_{i \in {1, \ldots, b}}V_i \cup \{t\} \\
			A_1 &\leftarrow \{(s, s')\} \\ 
			A_2 &\leftarrow \{(s', (m_i, j))| i \in {1, \ldots, m}, j \in {1, \ldots, n} \} \\
			A_3 &\leftarrow \{((m_i, j), m_i)| i \in {1, \ldots, m}, j \in {1, \ldots, n} \} \\
			A_4 &\leftarrow \{(m_i, (m_i, V_j))| i \in {1, \ldots, m}, j \in {1, \ldots, b}\} \\
			A_5 &\leftarrow \{((m_i, V_j), v_j)|i \in {1, \ldots, m}, j \in {1, \ldots b}, v_j \in V_j\} \\
			A_6 &\leftarrow \cup_{i = 1}^{b} V_i \times \{t\} \\
			A   &\leftarrow A_1 \cup A_2 \cup A_3 \cup A_4 \cup A_5 \cup A_6 \\
			capacity(e) &= \left\{
			\begin{array}{ll}
			c \: |& \: e = (s, s')\\
			1 \: |& \: otherwise
			\end{array}
			\right. \\
			cost(e) &= \left\{
			\begin{array}{ll}
			j \:      &| \: e = (s', (m_i, j)) \\
			p_{j}^{i}-1 \: &|\: e = ((m_i, V_j), v_j) \land v_j \in V_j \\
			0 \:      &|\: otherwise
			\end{array}
			\right.
			\end{align*} 
			\STATE Calculate a maximum integral flow with the minimum cost in $D = (V, A, capacity, cost)$.
			\STATE Schedule $S$: assign the jobs according to the maximum flow in $D$.  
			\RETURN $S$.
		\end{algorithmic}
		\caption{An optimal algorithm for $R|\conflictgraph{}, p_{j}^{i} \in \{1, \infty \}|\sum C_j$.}
		\label{alg:identical_unit_time_restricted_assignment}
	\end{algorithm}
	\noindent
	\noindent 
	We leave the proof of the following claim to the reader.
	\begin{claim}
		\label{claim:OnSpreadingTheSmallJobs}
		Assume that for the problem $R|V_1 \cup V_2, p_j \in \{p_1 \le p_2\}|\sum C_j$ there is a schedule with $n_1$ jobs assigned with processing time $p_1$ and where the jobs are assigned according to Smith's Rule~\cite{SmithRule}.
		If the number of the machines to which these $n_1$ jobs are assigned is maximal, then it has the smallest total completion time among all the schedules with $n_1$ jobs assigned with processing time $p_1$.
	\end{claim}
	
	Let us guess the number of jobs assigned with processing time $p_1$ in a solution to an instance of $R|\conflictgraph[2], p_{j}^{i} \in \{p_1 \le p_2\}|\sum C_j$.
	By the claim and the algorithm following from \cref{theorem:SchedulingByFlows} we may find distribution of these jobs to the machines and schedule the rest of the jobs with processing time $p_2$.
	Hence we have the following.
	\begin{theorem}
		\label{theorem:TwoCliquesTwoTimes}
		\Copy{theorem:TwoCliquesTwoTimes:text}
		{
			$R|\conflictgraph[2], p_{j}^{i} \in \{p_1 \le p_2\}|\sum C_j$ can be solved in $O(m^2n^4\log mn)$ time.
		}
	\end{theorem}
	\begin{proof}
		Notice that we can guess the number of the jobs to be scheduled with processing time $p_1$.
		Let this number be $n_1$.
		Notice that any machine can do at most $2$ jobs.
		
		Also notice that if we have a partial schedule for some of the jobs, then we may always schedule the remaining jobs greedily with processing time at most $p_2$, respecting the cliques, and not increasing the completion time of the jobs already scheduled.
		To see that it is always possible, assume that there is a job $v \in V_i$, where $V_i$ is one of two cliques, such that it cannot be scheduled.
		Hence any machine on which it cannot be scheduled has to have a job from $V_i$ but the number of such machines clearly has to be less than $m$.
		
		Let us use \tpikiesalgorithm{\ref{alg:identical_unit_time_restricted_assignment}} to find a subset of $n_1$ jobs that is assigned to as many machines as possible - let us set the parameter $c$ to $n_1$ in the algorithm.
		Hence, by an easy observation, the optimal flow corresponds to a schedule of $n_1$ jobs assigned with processing time $p_1$, with the minimum total completion time among all such schedules.
		Clearly, greedy assignment of the remaining jobs is optimal - and by the fact that we guessed the number of jobs assigned with processing time $p_1$ we may assume that all of the remaining jobs are assigned with~$p_2$.
	\end{proof}
	
	Consider the problem $P|\conflictgraph[b], \bagmachinerestriction|\sum C_j$.
	Take any subset of cliques and order it, such an ordered subset we call a {\em configuration}.
	The number of configurations is $O(1)$, by the fact that $b$ is constant in this setting. 
	We may guess how many machines have a given configuration and we can check if all the guesses are feasible with respect to jobs.
	We may, by a matching method, check if we can assign the configurations to actual machines.
	After this by a matching technique similar to the one used in \cite{DBLP:journals/cacm/BrunoCS74} we may find the best schedule for a given multiset of configurations.
	Together this gives us the following.
	\begin{theorem}
		\label{theorem:identicalMachinesBBags}
		\Copy{theorem:identicalMachinesBBags:text}
		{
$P|\conflictgraph[b], \bagmachinerestriction|\sum C_j$ can be solved in time $O(m^{f(b)}n^{3}m^{}\log mn)$, where $f(b) = O(\sum_{i = 0}^{b} \binom{b}{i}i!)$.
		}
	\end{theorem}
	\begin{proof}
		Consider \tpikiesalgorithm{\ref{alg:IdenticalBagsRestricted}}.
		First notice that if the number of cliques is limited, then also the number of the possible assignments of the cliques to a machine is limited; we call such an assignment {\em machine configuration}.
		By an assignment of a clique we understand a reservation of a place for a job of the clique $V_i$ on a $k$-th position.
		The total possible number of ways to assigns cliques to a machine is $f(b) = \sum_{i = 0}^{b} \binom{b}{i}i!$, which corresponds to all the choices of $i=0,\ldots,b$ cliques among $b$ cliques in total and the assignment the chosen cliques to the positions $1, \ldots, i$ on the machine.
		Such an assignment done for $m$ machines at once we call {\em machines configuration}.
		Hence we check all possible machines configurations and their number is $O(m^{f(b)})$.
		Given a machines configuration it might be the case that the configuration has places reserved that are incompatible  with the jobs in $V_1 \cup \ldots \cup V_b$, in this case we may simply skip the configuration.
		Given a machines configuration it might be the case that due to clique dependent restrictions it is not possible to assign the machine configurations to machines.
		This is verified by finding the largest matching in the graph $(V, E)$.
		If there is perfect matching $M'$, then it is possible to assign machine configurations to the machines.
		The meaning of the matching $M''$ is that we have to assign the jobs from cliques to the positions in the configuration, represented by $(V', E', cost)$, which is a construction similar to the one presented in~\cite{DBLP:journals/cacm/BrunoCS74}.
		Hence using $M''$ and $M'$ one can easily construct a schedule.
		A feasible schedule with the smallest cost is an optimal one.
	\end{proof}
	
	\begin{algorithm}
		\begin{algorithmic}[1]
			\REQUIRE A set of cliques $V_1 \cup \ldots \cup V_b$, a set of $m$ identical machines $M$, a clique-machine compatibility graph $G_{bm}$.
			\ENSURE An optimal schedule or an information that the schedule does not exist.
			\FOR {machines configuration $MC$ in the set of all possible machines configurations}
			\STATE If $MC$ contains different number of places for cliques than $|V_1|, \ldots, |V_b|$ continue.
			\STATE Let $V = M \cup MC$.
			\STATE Let $E = \{\{m,C\}|m \in M, C \in MC, C=(V_a, \ldots, V_z), \{m,V_a\} \in G_{bm}, \ldots, \{m, V_z\} \in G_{bm}\}$.
			\STATE If there is no perfect matching in $G = (V, E)$ continue; otherwise let $M'$ be the matching.
			\STATE Let $V' = \bigcup_{C \in MC}\bigcup_{i = 1,\ldots, |C|} \{(C, i)\} \cup V_1 \cup \ldots \cup V_b$ 
			\COMMENT {By an abuse of the notation we assume that $C$ is an ordered subset of $\{1, \ldots, b\}$ corresponding to a machine configuration in $MC$, hence $C[i]$ is the $i$-th clique from the end in the configuration.}
			\STATE Let $E' = \bigcup_{C \in MC}\bigcup_{i = 1,\ldots, |C|} V_{C[i]} \times \{(C, i)\}$.
			\STATE Let $\forall_{v \in C[i]} cost(\{(C,i), v\}) = ip(v)$.
			\STATE Find the matching with the smallest cost $M''$ in $(V', E', cost)$.
			\STATE Schedule $S$: Assign jobs to machine configurations based on $M''$, assign machine configurations to machines based on $M'$.
			\ENDFOR
			\RETURN {Schedule with the smallest total completion time or "NO" if no feasible schedule was found.}
		\end{algorithmic}
		\caption{An exact algorithm for $P|\conflictgraph[b], \bagmachinerestriction|\sum C_j$.}
		\label{alg:IdenticalBagsRestricted}
	\end{algorithm}
	
	By a simple dynamic programming we obtain. 
	\begin{theorem}
		\label{theorem:RmConstantSizesNumber}
		\Copy{theorem:RmConstantSizesNumber:text}
		{
			$Rm |\conflictgraph{}, p_{j}^{i} \in \{a_1, \ldots, a_k\}| \sum C_j$ can be solved in $O(n^{2km}nm^m)$ time.
		}
	\end{theorem}
	\begin{proof}
		Consider \tpikiesalgorithm{\ref{alg:MUnrelatedConstantNumberOfValues}}.
		Notice that if the number of possible processing times is $k$, then each of the machines can have at most $O(n^{2k})$ jobs assigned with different processing times. 
		Hence the number of all possible divisions of the jobs can be bounded by $O(n^{2km})$.
		The algorithm processes the jobs clique by clique, each of the considered assignments does not contradict the cliques.
		Precisely, a clique consists of at most $m$ jobs that can be assigned to at most $m$ machines, hence the number of all possible proper assignments is $O(m^m)$.
		Notice that after considering the assignment of a clique $V_i$ ,the assignments of the jobs from cliques $V_1, \ldots, V_i$ that result in the same number of jobs with given size assigned to the machines are equivalent.
		Hence, we may discard all of them save one.
		This operation is the trimming step in the algorithm. 
		The trimming can clearly be done in $O(n^{2km})$ time.
		Hence the time complexity of the algorithm is $O(n^{2km}nm^m)$.
	\end{proof}
	\begin{algorithm}
		\begin{algorithmic}[1]
			\REQUIRE A set of cliques $V_1 \cup \ldots \cup V_b$ with jobs, a set of $m$ machines, a set $\{a_1, \ldots, a_k\}$ of possible values.
			\ENSURE An exact schedule.
			\STATE $divisions \leftarrow \{((0, \ldots, 0), \ldots, (0, \ldots, 0))\}$.
			\COMMENT {The single tuple in $divisions$ represents the number of jobs of sizes $a_1, \ldots, a_k$ on the machines $m_1 \ldots m_m$.}
			\FOR {$i = 1,\ldots, b$}
			\STATE $divisions' \leftarrow \emptyset$.
			\FOR {$d \in divisions$}
			\STATE Add the jobs in $V_i$ to $d$ in all possible ways.
			\STATE Add $d$ to $divisions'$.
			\ENDFOR
			\STATE Trim $divisions'$.
			\STATE $divisions \leftarrow divisions'$.
			\ENDFOR
			\RETURN {A schedule based on the division with the smallest cost}
		\end{algorithmic}
		\caption{ An exact algorithm for $Rm|\conflictgraph{}, p_{j}^{i} \in \{a_1, \ldots, a_k\}|\sum C_j$.}
		\label{alg:MUnrelatedConstantNumberOfValues}
	\end{algorithm}
	
	By constructing a suitable flow network, similar to the one used in \cite{DBLP:journals/cacm/BrunoCS74}, with the cliques requirement satisfied by a construction similar to the one presented in \cref{figure:AlgorithmUnitTimeBags} we obtain:
	\begin{theorem}
		\label{theorem:copiesOfJobs}
		\Copy{theorem:copiesOfJobs:text}
		{
			$R|\conflictgraph{}, \jobmachinerestriction, \bagmachinespeed|\sum C_j$ can be solved to optimality in $O(m^2\cdot n^4 \cdot \log mn))$ time.
		}
	\end{theorem}
	\begin{proof}
		Consider \tpikiesalgorithm{\ref{alg:RJobMachineRestricionBagMachineSpeed}}.
		The proof is based upon a straightforward observation that the constructed flow network has integral capacities, hence there exist an integral flow that has minimum cost.
		The flow network is a straightforward adaptation of the network presented in~\cite{DBLP:journals/cacm/BrunoCS74}.
		It is easy to see that a schedule corresponding to such a flow respects the cliques due to capacities of $A_2$.
		Also it respects the restrictions of the jobs by the composition of $A_1$.
		The complexity follows from the size of the network and \cref{theorem:MaximumFlowMinimumCost}.
	\end{proof}
	The theorem is only interesting because it shows that the problem of executing copies of given jobs reliably can be solved in polynomial time, even if the machines are unrelated and some copies cannot be executed on some machines.
	\begin{algorithm}
		\begin{algorithmic}[1]
			\REQUIRE A set of cliques $V_1 \cup \ldots \cup V_b$, a set of $m$ machines $M$, a mapping between machines and cliques, and relation $compatible \subseteq J \times M$ between jobs and machines.
			\ENSURE An optimal schedule
			\STATE Construct a flow network:
			\STATE Let there be sinks $T = M \times \{1, \ldots, n\}$, each with capacity $1$.
			\STATE Let there be sources $S = V_1 \cup \ldots \cup V_b$, each with capacity $1$.
			\STATE Let there be vertices $V^1 = M \times \{1, \ldots, b\} \times \{1\}$.
			\STATE Let there be vertices $V^2 = M \times \{1, \ldots, b\} \times \{2\}$
			\STATE $A_1 = \{(j, (m, i, 1))|j \in V_i, m \in M, (j, m) \in compatible \}$.
			\STATE $A_2 = \{((m,i,1), (m,i,2)) | m \in M, i \in \{1, \ldots, b\}\}$.
			\STATE $A_3 = \{((m,i,2), (m,n'))| n' \in \{1, \ldots, n\}, m \in M, i \in \{1, \ldots, b\}\}$.
			\STATE $capacity(e) \equiv 1$.
			\STATE 	$cost(e) = \left\{
			\begin{array}{ll}
			n'p_{k}^i \: &|\: e = ((i,k,2), (i,n')) \\
			0 \:      &|\: otherwise
			\end{array}
			\right\} $.
			\\
			\COMMENT{By an abuse of the notation, we assume that for a clique $V_k$, $p_{k}^i$ is the processing time of a job from $V_k$ on $m_i$.}
			\STATE Construct the maximum flow with minimal cost in $(S \cup T \cup V^1 \cup V^2, A_1 \cup A_2 \cup A_3, capacity, cost)$. 
			\RETURN {If the flow is less than $n$, then there is no feasible schedule. 
				Otherwise return a schedule corresponding to the flow.}
			
		\end{algorithmic}
		\caption{An exact algorithm for $R|\conflictgraph{}, \jobmachinerestriction,\bagmachinespeed|\sum C_j$.}
		\label{alg:RJobMachineRestricionBagMachineSpeed}
	\end{algorithm}
	
	
	\section{Hardness Results}\label{sec:hardness}
	In this chapter we prove hardness results for the following problems:
	$R|\conflictgraph[2], p_{j}^{i} \in \{p_1 < p_2 < p_3\}|\sum C_j$;
	$P|\conflictgraph{}, p_j \in \{p_1 < p_2\},\bagmachinerestriction|\sum C_j$, where each of the cliques has at most $2$ jobs and
	$R|\conflictgraph{}, p_j \in \{p_1 < p_2\}|\sum C_j$.
	We do this by modifying results from \cite{DBLP:conf/stacs/MaackJ20}.
	Similar techniques have been used before, see, e.g., \cite{DBLP:journals/algorithmica/EbenlendrKS14,DBLP:journals/jcss/ChenJZ18,DBLP:conf/stacs/0011MYZ17}.
	We also prove that $P|\conflictgraph{}, \bagmachinerestriction, p_i~\in~\{p_1 < p_2 < 2p_1\}|\sum C_j$ is \APXH by an $L$-reduction from the problem \MAXThreeSATSix.
	
	Let us start with a description of \MAXThreeSATSix.
	This problem is an optimization version of \ThreeSat in which every variable appears in $6$ clauses and each literal in exactly $3$ clauses. 
	The goal is to calculate the maximum number of clauses that can be satisfied, i.e., have at least one literal with truth assigned.
	From the sketch of the proof of Theorem 12 from \cite{ApproximatingMinSumSetCover} we get the following lemma. 
	\begin{lemma}[\cite{ApproximatingMinSumSetCover}]
		The problem \MAXThreeSATSix is \APXH.
	\end{lemma}
	For the $L$-reduction let us use the definition from~\cite{DBLP:books/lib/Ausiello99}.
	Let $P_1$ and $P_2$ be two \NPOClass problems.
	This class consists of optimization problems such that for any problem in this set:
	\begin{itemize}
		\item The set of instances of the problem is recognizable in polynomial time.
		\item The value of a solution for the problem can be bound by a polynomial function of the size of an instance.
		Moreover, any such a solution can be verified in polynomial time to be a solution to the instance of the problem.
		\item The value of a solution can be computed in polynomial time.  
	\end{itemize}
	The terms used in the further definitions are as follows. 
	$I_{P_1}$ ($I_{P_2}$) is the set of instances of $P_1$ ($P_2$).
	$SOL_{P_1}(x)$ ($SOL_{P_2}(x)$) is the function that associates to any input instance $x \in I_{P_1}$ ($x \in I_{P_2}$) the set of feasible solutions of $x$.
	A function $m_{P_1}(x)$ ($m_{P_2}(x)$) is the measure function, defined for pairs $(x, y)$ such that $x \in I_{P_1}$ and $y \in SOL_{P_1}$ ($x \in I_{P_2}$ and $y \in SOL_{P_2}$);
	for every such pair $(x, y)$, $m_{P_1}(x)$ ($m_{P_2}(x)$) provide a positive integer (or rational) which is the value of the feasible solution $y$.
	Finally, $m^*_{P_1}(x)$ ($m^*_{P_2}(x)$) are the values of optimal solution for an instance $x \in I_{P_1}$ ($x \in I_{P_2}$).
	
	$P_1$ is said to be $L$-reducible to $P_2$ if functions $f$ and $g$ and two positive constant $\beta$ and $\gamma$ exist and are such that:
	\begin{enumerate}
		\item For any instance $x \in I_{P_1}$, $f(x) \in I_{P_2}$ is computable in polynomial time.
		\item For any $x \in I_{P_1}$, if $SOL_{P_1}(x) \neq \emptyset$, then $SOL_{P_2}(f(x)) \neq \emptyset$
		\item For any $x \in I_{P_1}$ and for any $y \in SOL_{P_2}(f(x))$, $g(x,y) \in SOL_{P_1}(x)$ is computable in polynomial time.
		\item For any $x \in I_{P_1}$, we have $m^*_{P_2}(f(x)) \le \beta m^*_{P_1}(x)$.
		\item For any $x \in I_{P_1}$ and for any $y\in SOL_{P_2}(f(x))$, we have
		\[
		|m^*_{P_1}(x) - m_{P_1}(x, g(x,y))| \le \gamma |m^*_{P_2}(f(x)) - m_{P_2}(f(x), y)|
		\]
	\end{enumerate}
	
	\begin{figure}
		\centering
		\small
		\begin{tikzpicture}
		[place/.style={ellipse,draw=black!100,line width=0.3mm,inner sep=0pt,minimum size=6mm},
		jobsB1/.style={ellipse,draw=black!100,line width=0.3mm, dotted,inner sep=0pt,minimum size=6mm},
		jobsB2/.style={ellipse,draw=black!100,line width=0.3mm, dashed,inner sep=0pt,minimum size=6mm},
		jobsB1B2/.style={ellipse,draw=black!100,line width=0.3mm, dash dot,inner sep=0pt,minimum size=6mm},
		myline/.style={line width=0.3mm}]
		\node at ( 0,0.25) (j1) [place] {$j[v,1]$};
		\node at ( 2,0.25) (j2) [place] {$j[v,2]$};
		\node at ( 4,0.25) (j3) [place] {$j[v,3]$};
		\node at ( 6,0.25) (j4) [place] {$j[v,4]$};
		\node at ( 8,0.25) (j5) [place] {$j[v,5]$};
		\node at ( 10,0.25) (j6) [place] {$j[v,6]$};

		\node at ( 0,-1) (m1) [place] {$m[v,1]$};
		\node at ( 2,-1) (m2) [place] {$m[v,2]$};
		\node at ( 4,-1) (m3) [place] {$m[v,3]$};
		\node at ( 6,-1) (m4) [place] {$m[v,4]$};
		\node at ( 8,-1) (m5) [place] {$m[v,5]$};
		\node at ( 10,-1) (m6) [place] {$m[v,6]$};
		
		\node at ( 0,-2) (j1T) [place] {$j^T[v,1]$};
		\node at ( 2,-2) (j1F) [place] {$j^F[v,1]$};
		
		\node at ( 2,-3) (mC1) [place] {$m[C,1]$};
		\node at ( 5,-3) (mC2) [place] {$m[C,2]$};
		\node at ( 8,-3) (mC3) [place] {$m[C,3]$};
		
		\node at ( 2,-4) (jC1) [place] {$j[C,1]$};
		\node at ( 5,-4) (jC2) [place] {$j[C,2]$};
		\node at ( 8,-4) (jC3) [place] {$j[C,3]$};

		\draw [myline,->] (j1.south) --(m1.north);
		\draw [myline,->] (j1.south) --(m2.north west);
		\draw [myline,->] (j2.south) --(m2.north);
		\draw [myline,->] (j2.south) --(m3.north west);
		\draw [myline,->] (j3.south) --(m3.north);
		\draw [myline,->] (j3.south) --(m4.north west);
		\draw [myline,->] (j4.south) --(m4.north);
		\draw [myline,->] (j4.south) --(m5.north west);
		\draw [myline,->] (j5.south) --(m5.north);
		\draw [myline,->] (j5.south) --(m6.north west);
		\draw [myline,->] (j6.south) --(m6.north);
		\draw [myline,->] (j6.south) to[out =195, in=15] (m1.north east);
		
		\draw [myline,->] (j1T.north) --(m1.south);
		\draw [myline,->] (j1F.north) --(m1.south east);
		
		\draw [myline,->] (j1T.south) --(mC1.north);
		
		\draw [myline,->] (jC1.north) --(mC1.south);
		\draw [myline,-] (jC1.north) |- ($(jC2.north)!.6!(mC2.south)$);
		\draw [myline,-] (jC1.north) |- ($(jC3.north)!.6!(mC3.south)$);
		
		\draw [myline,-] (jC2.north) |- ($(jC1.north)!.6!(mC1.south)$);
		\draw [myline,->] (jC2.north)--(mC2.south); 
		\draw [myline,-] (jC2.north) |-($(jC3.north)!.6!(mC3.south)$);
		
		\draw [myline,-] (jC3.north)  |- ($(jC1.north)!.6!(mC1.south)$);
		\draw [myline,-] (jC3.north)  |-($(jC2.north)!.6!(mC2.south)$);
		\draw [myline,->] (jC3.north) --(mC3.south);
		\end{tikzpicture}
		\caption
		{
			An illustration of the idea of eligibility of the jobs used in \cref{theorem:apxcompleteness}. 
			The figure presents a component corresponding to one of the variables and a component corresponding to one of the clauses.
			In the example $C$ is such a clause that $(C,1) = \kappa(v,1)$.
		}
		\label{figure:ComponentForBasicNPC}
	\end{figure}
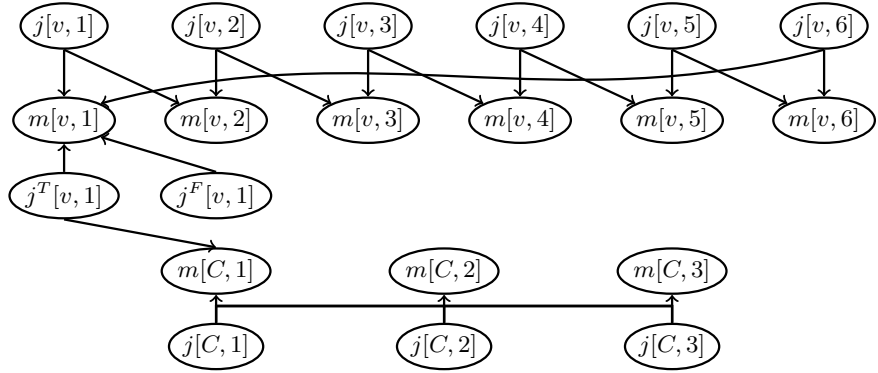
	
	\begin{table}
		\caption{The processing times $p_{i}$ of jobs used in the $L$-reduction in \cref{theorem:apxcompleteness}.}
		\centering
		\begin{tabular}{ l c c c }
			job & clique & $p_i$ & clique allowed on \\
			\hline
			$j[v,1]$ & $V[v,1]$ & $p_1$ & $m[v,1], m[v,2]$ \\
			$j[v,2]$ & $V[v,2]$ & $p_2$ & $m[v,2], m[v,3]$ \\
			$j[v,3]$ & $V[v,3]$ & $p_1$ & $m[v,3], m[v,4]$ \\
			$j[v,4]$ & $V[v,4]$ & $p_2$ & $m[v,4], m[v,5]$ \\
			$j[v,5]$ & $V[v,5]$ & $p_1$ & $m[v,5], m[v,6]$ \\
			$j[v,6]$ & $V[v,6]$ & $p_2$ & $m[v,6], m[v,1]$ \\
			\hline
			$j^T[v,i]$ & $V^*[v,i]$ & $p_1$ & $m[v,i], m[\kappa(v,i)]$\\
			$j^F[v,i]$ & $V^*[v,i]$ & $p_2$ & $m[v,i], m[\kappa(v,i)]$ \\
			\hline
			$j[C, 1]$   & $V[C,1]$ & $p_1$ & $m[C,1], m[C,2], m[C,3]$ \\
			$j[C, 2]$   & $V[C,1]$ & $p_1$ & $m[C,1], m[C,2], m[C,3]$ \\
			$j[C, 3]$ & $V[C,1]$ & $p_2$ & $m[C,1], m[C,2], m[C,3]$ \\   
			\\
		\end{tabular}
		\label{table:ProcessingTimesAPXCompleteness}
	\end{table}
	
	\begin{theorem}
		$P|\conflictgraph{}, \bagmachinerestriction, p_i \in \{p_1 < p_2 < 2p_1\}|\sum C_j$ is \APXH.
		\label{theorem:apxcompleteness}
	\end{theorem}
	\begin{proof}
		We prove this by an $L$-reduction from \MAXThreeSATSix to $P|\conflictgraph{},\bagmachinerestriction , p_i \in \{p_1 < p_2 < 2p_1\}|\sum C_j$.
		
		For the pair of the problems let us define $f$, the function constructing an instance of $P|\conflictgraph{}, \bagmachinerestriction, p_i \in \{p_1 < p_2 < 2p_1\}|\sum C_j$  from an instance of \MAXThreeSATSix.
		Let the set of variables be $V$; and the set of clauses be $\mathcal{C}$, where $|\mathcal{C}| = 2|V|$.
		Define $\kappa: V \times \{1, \ldots, 6\} \rightarrow \mathcal{C} \times \{1, 2, 3\}$ to be a function that maps the first unnegated literal of a variable, the first negated literal of the variable, etc. to its clause and the position in the clause.
		For a variable $v \in V$, construct a set of machines $\cup_{i = 1,\ldots, 6}\{m[v, i]\}$.
		The machine $m[v,1]$ corresponds to the first non-negated literal of $v$, $m[v,2]$ corresponds to first negated one, etc. 
		Construct also a set of clause machines $\{m[C, 1], m[C, 2], m[C, 3]\}$, for $C \in \mathcal{C}$. 
		The jobs that we construct are described in \cref{table:ProcessingTimesAPXCompleteness}.
		Notice that there are $13|V|$ jobs with size $p_1$ and $11|V|$ jobs with size~$p_2$.
		The construction is illustrated in \cref{figure:ComponentForBasicNPC}.
		
		Let $k$ be the maximum number of clauses that can be satisfied for a given instance of \MAXThreeSATSix.
		Notice that $|V| \le k \le 2|V|$, because if we assign $T$ to all the variables, then at least half of the clauses are satisfied.
		Let us make an assignment of the jobs to machines based on a valuation giving $k$ satisfied clauses.
		Consider two cases.
		\begin{itemize}
			\item If a variable $v$ has value $T$, let $m[v, 1], m[v, 3], m[v, 5]$ be assigned jobs $j[v, 1], j[v, 3], j[v, 5]$ and let $m[v, 2], m[v, 4], m[v, 6]$ be assigned jobs $j[v, 2], j[v, 4], j[v, 6]$. 
			\item Otherwise let $m[v, 1], m[v, 3], m[v, 5]$ be assigned jobs $j[v, 6], j[v, 2], j[v, 4]$ and let $m[v, 2],$ $m[v, 4], m[v, 6]$ be assigned jobs $j[v, 1], j[v, 3], j[v, 5]$. 
		\end{itemize}
		If $m[v, i]$ has job with processing time $p_2$ assigned already, assign a job with processing time $p_1$ from $V^*[v, i]$ to it; otherwise assign a job with processing time $p_2$ from $V^*[v, i]$ to it.
		Assign the other job from $V^*[v,i]$ to $m[\kappa(v,i)]$.
		For all $C \in \mathcal{C}$ assign the jobs from the clique $V[C,1]$ to the eligible machines in an optimal way.
		Notice that only the machines that correspond to the clauses that are not satisfied can have two jobs with size $p_2$ assigned, and there is exactly one such machine for a given not satisfied clause.
		Notice that the cost of such a schedule is
		\begin{gather*}
		6|V|(2p_1 + p_2) + (2|V| - k)(4p_1 + 5p_2) + (11|V| - 6|V| - 4(2|V| - k))(2p_1 + p_2) \\
		+ \frac{1}{2}(13|V| - 6|V| - 2(2|V| - k) - (11|V| - 6|V| - 4(2|V| - k)))3p_1 \\
		= 25|V|p_1 + 11|V|p_2 + (2|V| - k)(p_2-p_1) \le k(24p_1 + 12p_2).
		\end{gather*}
		Hence let $(24p_1 + 12p_2)$ be the $\beta$ constant.
		
		Let us assume that for a given instance of \MAXThreeSATSix we have a solution $y$ of the corresponding scheduling problem with a given cost. 
		Let us define the $g$ function.
		The $g$ function begins with modifying the solution according to the following observations.
		\begin{enumerate}
			\item Let us assume that in $y$ there exists $m[v, i]$ that has exactly $2$ jobs assigned; let us assume that both of them have size $p_1$ (have size $p_2$).
			Notice that this means that the machine has a job $j^T[v, i]$ (a job $j^F[v, i]$) assigned. 
			Notice that we can exchange this job with $j^F[v, i]$ (with $j^T[v, i]$) without increasing the total completion time.
			\item Assume that some machine $m[v, i]$ has three jobs assigned. 
			It also means that there is a machine $m[v, i']$ that has exactly one job assigned.
			Notice that in any case, by the previous observation and by the assumption that $p_1 \le p_2 \le 2p_1$ we may shift the jobs in a way that after the shift all of the machines have exactly $2$ jobs, without increasing the total completion time of the schedule.
			This follows from a simple analysis of all possible cases of the assignment of the jobs to the machines. 
		\end{enumerate}
		Notice that this means that we may assume that the machines $m[v, i]$ are processing exactly one job with size $p_1$ and one with size $p_2$ each. 
		We prove that the total completion time of the schedule depends only on the number of the machines that are processing two jobs with size $p_2$.
		Let the number of such machines be $k'$.
		Total completion time of the schedule is then equal to $k'3p_2 + (11|V| - 2k')(2p_1 + p_2) + \frac{1}{2}(13|V| - (11|V| - 2k'))3p_1 = 25|V|p_1 + 11|V|p_2 + k'(p_2-p_1)$.
		From such a schedule we can easily find a valuation of the variables in the instance of \MAXThreeSATSix such that it satisfies exactly $2|V| - k'$ clauses.
		Let now $k''$ be the number of machines that are processing two jobs with size $p_2$ in an optimal solution.
		Notice that $k''$ corresponds to a schedule with cost $25|V|p_1 + 11|V|p_2 + k''(p_2-p_1)$. 
		And this schedule corresponds to a solution to \MAXThreeSATSix that has exactly $(2|V| - k'')$ clauses satisfied.
		There can be no better solution to \MAXThreeSATSix.
		Hence let us assume that for some $\gamma$ we have that 
		\[
		|(2|V|  - k'') - (2|V| -k')| \le \gamma |k''(p_2-p_1) + 25|V|p_1 + 11|V|p_2 - (k'(p_2 - p_1) + 25|V|p_1 + 11|V|p_2)|.
		\]
		Which is equivalent to $k' - k'' \le \gamma (k' - k'')(p_2 - p_1)$,
		hence clearly $\gamma = \frac{1}{p_2 - p_1}$ is a suitable constant.
		All other conditions are easily fulfilled.
	\end{proof}
	The APX-hardness for $R|\conflictgraph{}, p_{j}^{i} \in \{p_1 < p_2 < p_3\} |\sum C_j$ follows readily from the observation that we may always set $p_3$ to such a high value (dependent on the size of an instance of the problem) that in any reasonable schedule it will be not used.
	Mind the difference with the previous problem, that in this case $p_3$ is a function of $p_1, p_2$ and the input size.
	
	The same idea may be reused for the next problem and an $\alpha$-reduction, but this time from an even more restricted version, i.e., from the problem \ThreeSatStar considered in \cite{DBLP:conf/stacs/MaackJ20}.
	The input of \ThreeSatStar problem consists of a set of variables, and two sets of clauses: 1-in-3 clauses and 2-in-3 clauses.
	Each of the literals occurs exactly $2$ times, hence each variable occurs exactly twice negated and twice nonnegated.
	The number of 1-in-3 clauses and 2-in-3 clauses are equal.
	The question is if there is assignment of the variables such that in each 1-in-3 clause exactly one literal is true and that in each 2-in-3 clause exactly two literals are true.
	In the paper it was proved that the problem is \NPC.
	
	In the case of the next problem we use $p_3$ to restrict assignment of some jobs to some machines. 
	We have to also divide the jobs differently.
	\begin{theorem}
		\label{theorem:R2BNPC}
		\Copy{theorem:R2BNPC:text}
		{
			$R|\conflictgraph[2], p_{j}^{i} \in \{p_1 < p_2 < p_3\}|\sum C_j$ is strongly \NPC.
		}
	\end{theorem}
	
	\begin{table}
		\caption{The processing times $p_{j}^{i}$ used in the $\alpha$-reduction in \cref{theorem:R2BNPC}.}
		\centering
		\begin{tabular}{ l c c c c }
			job & clique & $p_1$ on & $p_2$ on & $p_3$ on  \\
			\hline
			$j[v,1]$ & $V_1$ & $m[v,1], m[v,2]$ & - & other \\
			$j[v,2]$ & $V_1$ & - & $m[v,2], m[v,3]$ & other \\
			$j[v,3]$ & $V_1$ &$m[v,3], m[v,4]$ & - & other \\
			$j[v,4]$ & $V_1$ &  - & $m[v,4], m[v,1]$ & other \\
			\hline
			$j^T[v,i]$ & $V_2$ & $m[v,i], m[\kappa(v,i)]$ & - & other \\
			$j^F[v,i]$ & $V_2$ & - & $m[v,i], m[\kappa(v,i)]$ & other \\
			\hline
			$j[C, 1]$ & $V_1$ & $m[C,1], m[C,2], m[C,3]$ & - & other \\
			$j[C, 2]$ & $V_1$  & - & $m[C,1], m[C,2], m[C,3]$ & other \\
			\multirow{2}{*}{$j[C, 3]$} & \multirow{2}{*}{$V_1$}   & if $C \in C_{13}$: &  if $C \in C_{23}$: & other \\
			&  &      $m[C,1], m[C,2], m[C,3]$  & $m[C,1], m[C,2], m[C,3]$ & other \\ 
			\\  
		\end{tabular}
		\label{table:ProcessingTimes}
	\end{table}
	\label{apx:hardness}
	\begin{proof}
		Consider the proof of Proposition $9$ from~\cite{DBLP:conf/stacs/MaackJ20}.
		We encode the \ThreeSatStar problem as an instance of $R|\conflictgraph[2], p_{j}^{i} \in \{p_1 < p_2 < p_3\}|\sum C_j$.
		For an instance of \ThreeSatStar let $V$ be the set of variables, $C_{13}$ the set of 1-in-3 clauses and $C_{23}$ the set of 2-in-3 clauses.	
		Let $\kappa : V \times \{1,2,3,4\} ]\rightarrow C \times \{1,2,3\}$, i.e., let it be a function mapping respectively the first nonnegated, first negated, second nonnegated, second negated literal corresponding to $v \in V$ to a clause $C$ and a position in the clause.  
		
		We construct the following sets:
		sets of machines and jobs that correspond to variables,
		sets of machines and jobs that correspond to clauses,
		sets of jobs that force the valuation of literals in the clauses to be consistent with the valuation of variables.
		\begin{itemize}
			\item For a given variable $v \in V$ we construct machines $m[v,1], m[v,2], m[v,3], m[v,4]$ called {\em variable machines} and jobs $j[v,1], j[v,2], j[v,1], j[v,2]$ called {\em variable jobs}.
			\item For a given clause   $C \in C_{13} \cup C_{23}$ we construct machines $m[C,1], m[C,2], m[C,3]$ called {\em clause machines} and jobs $j[C, 1], j[C,2], j[C,3]$ called {\em clause jobs}.
			\item In addition we construct for each variable $v \in V$ jobs $j^T[v,1], j^T[v,2], j^T[v,3], j^T[v,4]$ and $j^F[v,1], j^F[v,2], j^F[v,3], j^F[v,4]$ called {\em consistency jobs}.
		\end{itemize}
		The cliques have two functions: they force consistency of the valuation of the literals; they also force that 1-in-3 clauses and 2-in-3 clauses are satisfied by literals, with consistency jobs acting as intermediaries.
		Notice that the total number of machines is $m = 8|V|$ and the total number of jobs is $n = 16|V|$.
		The processing times are given in \cref{table:ProcessingTimes}.
		The question is if there is a schedule with the total completion time equal to $m(2p_1 + p_2)$, which corresponds to a schedule where every job is scheduled with the lowest possible processing time and every machine has one job assigned with processing time $p_1$ and one with $p_2$.
		In fact the processing time $p_3$ is used to exclude some assignments, because in a schedule that has the required processing time each of the jobs have to be assigned with processing time $p_1$ or $p_2$.
		
		Assume that there is an assignment satisfying the \ThreeSatStar instance.
		Construct the schedule in the following way:
		\begin{enumerate}
			\item 
			If $v \in V$ has value true:
			\begin{itemize}
				\item $m[v,1] \leftarrow \{j[v,1], j^F[v,1]\}$, $ m[v,2] \leftarrow \{j[v,2], j^T[v,2] \}$,\\ $m[v,3] \leftarrow \{j[v,3], j^F[v,3]\}$, $m[v,4] \leftarrow \{j[v,4], j^T[v,4]\}$.
				\item $m[\kappa(v,1)] \leftarrow j^T[v,1], m[\kappa(v,2)] \leftarrow j^F[v,2], m[\kappa(v,3)] \leftarrow j^T[v,3], m[\kappa(v,4)] \leftarrow j^F[v,4]$.
			\end{itemize} 
			If $v \in V$ has value false:
			\begin{itemize}
				\item $m[v,1] \leftarrow \{j[v,2], j^T[v,1]\}$, $m[v,2] \leftarrow \{j[v,3], j^F[v,2] \}$, \\ $m[v,3] \leftarrow \{j[v,4], j^T[v,3]\}$, $m[v,4] \leftarrow \{j[v,1], j^F[v,4]\}$.
				\item $m[\kappa(v,1)] \leftarrow j^F[v,1], m[\kappa(v,2)] \leftarrow j^T[v,2], m[\kappa(v,3)] \leftarrow j^F[v,3], m[\kappa(v,4)] \leftarrow j^T[v,4]$.
			\end{itemize}
			\item 
			For a clause $C$ assign the jobs from $j[C,1], j[C,2], j[C,3]$ optimally.
			That is, assign them in a way that is consistent with the assignment performed in the previous step and in a way that each of the machines has one job with processing time $p_1$ and one with $p_2$.
			Notice that for a $C \in C_{13}$ the machines $m[C,1], m[C,2], m[C,3]$ have exactly two consistency jobs with processing time $p_2$ and one with $p_1$ assigned, hence it is always possible.
			Similar considerations hold for clauses $C \in C_{23}$.
		\end{enumerate}
		In such a schedule each of the machines have exactly one job with processing time $p_1$ and one with $p_2$.
		
		Now assume that there is a schedule $S$ with the total completion time equal to $m(2p_1 + p_2)$.
		Consider an assignment: for a variable $v$ if $m[v,1]$ has $j[v,1]$ assign $T$ to $v$, otherwise assign $F$.
		Notice that due to the processing times, $m[v,1]$ has $j[v,1]$ and $m[v,3]$ has $j[v,3]$ assigned; or $m[v,1]$ has $j[v,2]$ and $m[v,3]$ has $j[v,4]$ assigned.
		The jobs on the machines hence correspond to the $T/F$ values of the literals.
		In the schedule the jobs $j^T[v,i]$ and $j^F[v,i]$ have to complement the assignment of $j[v,i]$, hence the valuation of the "appearances" of the literals has to be also consistent. 
		Finally notice that due to the total completion time bound and due to the processing times, the machines $m[C,1], m[C,2], m[C,3]$ have exactly two jobs: $j^T[v,i]$ and $j^T[v', i']$ and one job $j^F[v'', i'']$ assigned if clause is $C \in C_{23}$.
		Similar observation holds for $C \in C_{13}$.
	\end{proof}
	By similar constructions we obtain the two following theorems.
	\begin{theorem}
		\label{theorem:BagMachineNPC}
		\Copy{theorem:BagMachineNPC:text}
		{
			$P|\conflictgraph{}, p_j \in \{p_1 < p_2\},\bagmachinerestriction|\sum C_j$ is strongly \NPC even if each clique has at most $2$ jobs.	
		}
	\end{theorem}
	\begin{proof}
		\label{theorem:BagMachineNPC:proof}
		We proceed similarly as in the proof of \cref{theorem:R2BNPC}.
		In fact we construct the same set of machines and the same set of jobs.
		However, we do not use $p_3$.
		The reasons are that clique restrictions are used instead, and that we form cliques differently, see \cref{table:ProcessingTimesRestricedAssignment}.
		As previously the total completion time limit is $m(2p_1 + p_2)$.
		Notice that the limit on the completion time forces that each of the components corresponding to a variable has to get exactly two jobs $j^T[v,i]$ and $j^T[v,i']$ and two jobs $j^F[v,i'']$ and $j^F[v,i''']$.
		This forces the structure of the assignment to correspond to a proper solution to \ThreeSatStar.
	\end{proof}
	\begin{table}
		\caption{The processing times $p_{j}^{i}$ used in the $\alpha$-reduction in \cref{theorem:BagMachineNPC}.}
		\centering
		\begin{tabular}{ l c c c }
			job & clique & $p_i$ & clique allowed on \\
			\hline
			$j[v,1]$ & $V[v,1]$ & $p_1$ & $m[v,1], m[v,2]$ \\
			$j[v,2]$ & $V[v,2]$ & $p_2$ & $m[v,2], m[v,3]$ \\
			$j[v,3]$ & $V[v,3]$ & $p_1$ & $m[v,3], m[v,4]$ \\
			$j[v,4]$ & $V[v,4]$ & $p_2$ & $m[v,4], m[v,1]$ \\
			\hline
			$j^T[v,i]$ & $V[v,5]$ & $p_1$ & $m[v,i], m[\kappa(v,i)]$\\
			$j^F[v,i]$ & $V[v,5]$ & $p_2$ & $m[v,i], m[\kappa(v,i)]$ \\
			\hline
			$j[C, 1]$ & $V[C,1]$ & $p_1$ & $m[C,1], m[C,2], m[C,3]$ \\
			$j[C, 2]$ & $V[C,2]$ & $p_2$ & $m[C,1], m[C,2], m[C,3]$ \\
			\multirow{2}{*}{$j[C, 3]$} & \multirow{2}{*}{$V[C,3]$} & if $C \in C_{13}$: $p_1$ & \multirow{2}{*}{$m[C,1], m[C,2], m[C,3]$}  \\
			&                          & if $C \in C_{23}$: $p_2$ &        \\  
			\\ 
		\end{tabular}
		\label{table:ProcessingTimesRestricedAssignment}
	\end{table}

	\begin{theorem}
		\label{theorem:2SizesNPC}
		\Copy{theorem:2SizesNPC:text}
		{
			$R|\conflictgraph{}, p_j \in \{p_1 < p_2\}|\sum C_j$ is strongly \NPC.
		}
	\end{theorem}
	\begin{proof}
		\label{theorem:2SizesNPC:proof}
		As previously, let there be an instance of \ThreeSatStar with set of variables $V$.
		In the case of this problem we construct many dummy jobs to emulate the restricted assignment.
		For clarity let us define a function $cost(x_1, x_2) =  \frac{x_2(x_2 + 1)}{2}p_2 + x_2x_1p_2 + \frac{x_1(x_1+1)}{2}p_1$, i.e., it is the total completion time of $x_2$ jobs with processing time $p_2$ and $x_1$ jobs with processing time $p_1$ scheduled on a single machine according to Smith's Rule.
		Consider the data in \cref{table:TwoProcessingTimes}, notice that there are $b = 7|V| + 4/3|V|$ cliques.
		The bound on the total completion time is $4|V|cost(b-2, 1) + 4|V|cost(b-1, 1)$.
		Notice that the bound corresponds to an assignment of the smallest possible number of jobs with processing time $p_2$ and maximal number of jobs with processing time $p_1$, moreover in a further described optimal way.
		By this property each of the clauses machines has to have $b$ jobs and each of variable machines has to have $b-1$ jobs.
		That is, half of the machines have to have $b-2$ jobs assigned with processing time $p_1$ and one job with processing time $p_2$, these have to be variable machines.
		The second half of the machines have to have $b-1$ jobs with processing time $p_1$ and $1$ with $p_2$, these are the machines corresponding to the clauses.  
		Notice that this forces the assignment of the variable jobs to variable machines to be consistent; the dummy jobs $j^*$, $j^{**}$ force this.
	\end{proof}
	\begin{table}
		\caption{The processing times $p_{j}^{i}$ used in the $\alpha$-reduction in \cref{theorem:2SizesNPC}.}
		\centering
		\begin{tabular}{ l c c c }
			job & clique & $p_1$ on & $p_2$ on \\
			\hline
			$j[v,1]$ & $V_1[v]$ & $m[v,1], m[v,2]$ & other \\
			$j[v,2]$ & $V_2[v]$ & - &  other \\
			$j[v,3]$ & $V_1[v]$ &$m[v,3], m[v,4]$ &  other \\
			$j[v,4]$ & $V_3[v]$ &  - & other \\
			$j[v,i]$, $i \in [5,m]$  & $V_1[v]$ & $M \setminus \{m[v,1], m[v,2], m[v,3], m[v,4] \}$ & $m[v,1], m[v,2], m[v,3], m[v,4]$ \\
			$j^*[v,i]$, $i \in [3,m]$  & $V_2[v]$ & $M \setminus \{ m[v,2], m[v,3]\}$ & $ m[v,2], m[v,3]$\\
			$j^{**}[v,i]$, $i \in [3,m]$  & $V_3[v]$ & $M \setminus \{m[v,1], m[v,4] \}$ & $m[v,1],  m[v,4]$ \\
			\hline
			$j^T[v,i]$ & $V[v,i]$ & $m[v,i], m[\kappa(v,i)]$ &  other \\
			$j^F[v,i]$ & $V[v,i]$ & - &  other \\
			$j^*[v,i,j], j \in [3,m]$ & $V[v,i]$ & $M \setminus \{m[v,i], m[\kappa(v,i)]\}$ & $m[v,i], m[\kappa(v,i)]$ \\
			\hline
			$j[C, 1]$ & $V[C]$ & $m[C,1], m[C,2], m[C,3]$ & other \\
			$j[C, 2]$ & $V[C]$  & - &  other \\
			$j[C, 3]$ & $V[C]$ & if $C \in C_{13}$:  $m[C,1], m[C,2], m[C,3]$  & other \\  
			$j^*[C, i], i \in [4, m]$ & $V[C]$ & $M \setminus m[C,1], m[C,2], m[C,3]$ & $m[C,1], m[C,2], m[C,3]$ \\  
			\\
		\end{tabular}
		\label{table:TwoProcessingTimes}
	\end{table}

\section{FPT Results}\label{sec:fptalgos}
	This section presents the FPT results for scheduling with clique incompatibility considering different parameterizations. To solve these problems, the algorithms model the respective problem as \nfold{} Integer Programs. 
	These IPs are of specific form: The constraint matrix consists of non-zero entries only in the first few rows and in blocks along the diagonal beneath. 
	Further we have to assure that the introduced objective functions are separable convex. Then the \nfold{} IP and thus the underlying problem can be solved efficiently. 
	The FPT results we obtain this way are:
	\begin{itemize}
		\item the problem $P|\conflictgraph{}, \bagmachinerestriction|\sum C_j$ can be solved in FPT time parameterized by the number of cliques $b$,
		\item the problem $R|\conflictgraph{}|\sum w_j C_j$ can be solved in FPT time parameterized by the number of machines $m$, the largest processing time $p_{\max}$ and the number of job kinds $\vartheta$,
		\item the problem $R|\conflictgraph{}|\sum w_j C_j$ can be solved in FPT time parameterized by the number of cliques $b$, the number of machine kinds $\kappa$, the largest processing time $p_{\max}$ and the number of job kinds $\vartheta$. 
	\end{itemize}
		The basis for the last two algorithms is formed by the work~\cite{DBLP:journals/scheduling/KnopK18} of Knop and Kouteck{\'{y}}. 
	Therein the authors prove FPT results for $R||\sum w_j C_j$ by formulating the problems as $n$-fold IPs with an appropriate objective function and similar parameters. 
	We prove that these IPs can be extended to handle clique incompatibility by carefully adapting the variables, the IPs and the objective functions, yielding the results above. Note that in \cite{DBLP:journals/corr/abs-1909-07326} these results are generalized, but by that also more complex. Further, using these results does not improve upon our running times.
	But first, let us give a short introduction to FPT and \nfold{} Integer Programming necessary to understand the following results. For details on FPT we refer to the standard textbook~\cite{DBLP:books/sp/CyganFKLMPPS15}. For details on \nfold{} IPs, we recommend~\cite{DBLP:journals/corr/abs-1904-01361}. 

	\subparagraph*{FPT.}
	In the parameterized complexity world a language is defined as $L\subseteq \{0,1\}^{*}\times \mathbb{N}$ where the first element encodes the instance and the second element, called \emph{parameter}, gives some further knowledge about the problem. This parameter may include the size of a solution, the treewidth of the graph, the number of variables in a formula, et cetera \cite{DBLP:books/sp/CyganFKLMPPS15}.
	A problem is \emph{fixed-parameter tractable} (FPT) if there is an algorithm that decides if $(x,k)\in L$ in time $f(k)\cdot |x|^c$ for a computable function $f$ and constant $c$.
	
	\subparagraph*{\nfold{} IP.} 
	Let $n,r,s,t \in \mathbb N$. Let $A_1,\dotsc,A_n \in \mathbb{Z}^{r \times t}$ and $B_1,\dotsc,B_n \in \mathbb{Z}^{s \times t}$ be integer matrices. The constraint matrix $\mathcal A \in \mathbb{Z}^{(r+n\cdot s) \times (n\cdot t)}$ of an \nfold{} IP is of following form:
	\begin{equation*}
	\mathcal A =
	\begin{pmatrix}
	A_1	& A_2	& \dots	& A_n      \\
	B_1	& 0 	& \dots  	& 0 	  \\
	0	&B_2	&\dots 	& 0	\\
	\vdots	& \vdots 	& \ddots & \vdots \\
	0 	& 0 & \dots 	 & B_n
	\end{pmatrix}.
	\end{equation*}
	Denote by $\Delta$ the largest absolute value in $\mathcal A$. We distinguish the constraints as follows: Denote the constraints (rows) corresponding to the $A_i$ matrices \textit{globally uniform} and the ones corresponding to the $B_i$ matrices \textit{locally uniform}.

	A function $g: \mathbb{R}^n \rightarrow \mathbb{R}$ is called separable convex if there exist convex functions $g_i:\mathbb{R} \rightarrow \mathbb{R}$ for each $i\in[n]$ such that $g(x)=\sum_{i=1}^ng_i(x_i)$.
	Let $f:\mathbb{R}^{nt}\rightarrow\mathbb{R}$ be some separable convex function and $b \in \mathbb{Z}^{r+n\cdot s}$. Further, denote by $\ell$ and $u$ some upper and lower bounds on the variables. 
	The corresponding \nfold{} Integer Program (\nfold{} IP)  is defined by $
	\min\,\{f(x) \ \vert\ \mathcal Ax = b, \ell \leq x \leq u,\, x \in \mathbb{Z}^{n\cdot t} \} $.
	The main idea for solving these IPs relies on local improvement steps which are used to converge from an initial solution to an optimal one yielding:
	\begin{proposition}[\cite{DBLP:journals/corr/abs-1904-01361}] \label{p:nfold}
		The Integer Program (\nfold{} IP) can be solved in time $(\Delta r s)^{O(r^2s+rs^2)}$ $nt \log(nt) \log(\|u-\ell\|_\infty)\log(f_{\max})$ where $f_{\max} = \max\sett[\big]{|f(x)|}{\ell\leq x\leq u}$. 
	\end{proposition}

	\subsection{Scheduling with Clique Machine Restrictions}
	We consider the problem variant $P|\conflictgraph{}, \bagmachinerestriction|\sum C_j$.
	Recall that in this setting we have a set $M(k)$ of machines for each clique $k\in [b]$. In a feasible schedule jobs of clique $k$ are scheduled exclusively on machines $i\in M(k)$.
	We prove the following result:
	
	\begin{theorem}
		The problem $P|\conflictgraph{}, \bagmachinerestriction|\sum C_j$ can be solved in FPT time parameterized by the number of cliques $b$.
	\end{theorem}
	
	To prove this result, we first establish some notation and basic observation, then introduce an Integer Programming model with $n$-fold form for the problem, and lastly argue that it can be solved efficiently.
	
	In any schedule for an instance of the problem there can be at most $b$ jobs scheduled on each machine due to the clique constraints.
	Hence, we may imagine that there are $b$ slots on each machine numbered in chronological order.
	We further use the intuition that the slots form $b$ layers with all the first slots in the first layer, all the second slots in the second one, and so on.
	Obviously, we can represent any schedule by an assignment of the jobs to these slots.
	Some of the slots may be empty, and we introduce the convention that all the empty slots (hence taking $0$ time) on a machine should be in the beginning. 
	If a job of clique $k$ is scheduled in a certain slot, we say that $k$ is present in the slot, in the corresponding layer and on the machine.
	In the following, we are interested in the pattern of cliques present on the machine and call such a pattern a \emph{configuration}.
	More precisely, we call a vector $C\in\set{0,1,\dots,b}^b$ a configuration if the following two conditions are satisfied:
	\begin{itemize}
		\item $\forall \ell,\ell'\in[b]: C_{\ell} = C_{\ell'} \wedge \ell\neq \ell' \implies C_{\ell}=C_{\ell'} = 0$
		\item $\forall \ell: C_{\ell} >0 \wedge \ell<b \implies C_{\ell+1} > 0$
	\end{itemize}
	Note that the $0$ represents an empty slot. The first condition corresponds to the requirement that at most one job of a clique should be scheduled on each machine. The second one matches to the convention that the empty slots are at the beginning.
	We denote the set of configurations as $\confs$.
	Moreover, $\confs(k)$ denotes for each $k\in [b]$ the set of configurations in which $k$ is present, i.e., $\confs(k) = \sett{C\in\confs}{\exists \ell\in [b]:C_\ell = k}$.
	Note that $|\confs| \leq (b+1)!$ since there can be up to $b$ zeros in a configuration and a configuration excluding the zeros can be seen as a truncated permutation of the numbers in $[b]$.
	We call a configuration $C$ eligible for a machine $i$ if all the cliques occurring in $C$ are eligible on $i$, that is, for each $C_{\ell}\neq 0$ we have $i\in M(C_{\ell})$.
	
	A schedule for an instance of the problem trivially induces an assignment of the machines to the configurations.
	We call such an assignment $\tau: M \rightarrow \confs$ feasible if there exists a feasible schedule corresponding to $\tau$.
	That is, if $\tau(i)$ is eligible on $i$ for each machine $i$ and, for each clique $k$, the number of machines assigned to a configuration in $\confs(k)$ is equal to the number of jobs in $k$.
	Obviously, different schedules may have the same assignment.
	However, we argue that given a feasible assignment $\tau$, we can find a schedule corresponding to $\tau$ with a minimal objective function value via a simple greedy procedure.
	Namely, for each clique $k$ we can successively choose a smallest job that is not scheduled yet and assign it to a slot positioned in the lowest layer that still includes non-empty slots belonging to $k$ according to $\tau$.
	Due to this observation, we can associate an objective value to each feasible assignment. In the next step we introduce an Integer Program to search for a feasible assignment $\tau$ with minimal objective.
	
	We introduce two types of variables, that is, $x_{C,i}\in\set{0,1}$ for each machine $i\in M$ and configuration $C\in\confs$ corresponding to the choice of whether $i$ is assigned to $C$ or not. Further, we have $y_{k,\ell} \in\set{0,1,\dots, n}$ for each clique $k\in[b]$ and layer $\ell\in [b]$ counting the number slots reserved for clique $k$ in the layers $1$ to $\ell$. 
	Moreover, we ensure $x_{C,i} = 0$ if $C$ is not eligible on $i$ using more restrictive upper bounds.
	Let $\mathcal{C}(k,\ell) = \sett{C\in\mathcal{C}}{\exists\ell\in[\ell]: C_{\ell'} = k}$ for each $k,\ell\in[b]$, $n_k $ be the number of jobs belonging to clique $k$, and $p_{k,s}$ the size of the job that has position $s$ if we order the jobs of clique $k$ non-decreasingly by size.
	Now the Integer Program has the following form:
	\begin{align}
	\min \sum_{\ell,k \in [b]}\sum_{s=1}^{y_{k,\ell}}& p_{k,s}&\nonumber\\
	\sum_{C\in \mathcal{C}} x_{C,i} &=1& \forall i\in M \label{eq:ip_bag_constraints_one_conf}\\
	\sum_{i\in M}\sum_{C\in\mathcal{C}(k,\ell)}x_{C,i}&=y_{k,\ell} & \forall k\in [b], \ell\in [b] \label{eq:ip_bag_constraints_x_and_y}\\
	y_{k,b} &= n_k & \forall k\in[b] \label{eq:ip_bag_constraints_jobs_covered}
	\end{align}
	Constraint (\ref{eq:ip_bag_constraints_one_conf}) ensures that exactly one configuration is chosen for each machine; due to (\ref{eq:ip_bag_constraints_x_and_y}), the variables $y_{k,\ell}$ correctly count the slots reserved for clique $k$; and (\ref{eq:ip_bag_constraints_jobs_covered}) guarantees that the jobs of each clique are covered.
	Finally, the objective function corresponds to the one described above:
	For each clique $k$, we sum up the smallest $y_{k,1}$ job sizes for the first layer, the smallest $y_{k,2}$ sizes in the second one, and so on.
	Note that this counting is correct since we use the convention that empty slots are at the bottom and therefore each job contributes once to the objective for its own layer and once for each layer above.
	Although the Integer Program does not have a linear objective and super-constant number of variables and constraints, we can solve it in suitable time using \nfold{} techniques:
	\begin{lemma}
		\label{lemma:IPRunningTime}
		\Copy{lemma:IPRunningTime:text}
		{
			The above IP can be solved in time $2^{O(b^4\cdot \log(b))}m\log(m)\log(n)\log(mp_{\max})$.
		}
	\end{lemma}
	\begin{proof}
		In order to use algorithms for \nfold{} IPs, we have to show that the IP has a suitable structure and the objective function is separable convex.
		
		To obtain the desired structure, we have to duplicate the $y$ variables for each machine.
		Hence, we get variables $y_{k,\ell, i}$ for each $i\in M$ and $ k,\ell\in[b]$. 
		We choose some machine~$i^*\in M$ and set $y_{k,\ell, i} = 0$ for each $i\neq i^*$ using lower and upper bounds for the variables.
		In the constraints (\ref{eq:ip_bag_constraints_x_and_y}) and (\ref{eq:ip_bag_constraints_jobs_covered}) we have to replace each occurrence of $y_{k,\ell}$ by $\sum_{i\in M} y_{k,\ell, i}$.
		Moreover, we have to change the objective to $\min \sum_{\ell,k \in [b]}\sum_{s=1}^{y_{k,\ell,i^*}} p_{k,s}$.
		It is easy to see that the resulting IP is equivalent and has an \nfold{} structure with one brick for each machine, a brick size of $t \leq b^2 + (b+1)!$, and a maximum absolute entry of $\Delta = 1$.
		Constraint (\ref{eq:ip_bag_constraints_one_conf}) is locally uniform, and the other constraints are globally uniform.
		Hence, we have $s = 1$ and $r = b^2 + b$.
		
		Concerning the objective function, first note that many of the variables do not occur in the objective and hence can be ignored in the following.
		We essentially have to consider the function $g_k: [n_k] \rightarrow \mathbb{R}, q\mapsto \sum_{s=1}^{q} p_{k,s}$ for each $k\in [b]$ since the objective can be written as $\sum_{\ell,k \in [b]}g_k(y_{k,\ell,i^*})$.
		Let $\set{x} = x - \floor{x}$ for each $x\in\mathbb{R}$ and $\tilde{g}_k:\mathbb{R}\rightarrow\mathbb{R}$ with:
		\[
		x\mapsto 
		\begin{cases}
		p_{k,1}x       & \quad \text{if } x < 1 \\
		p_{k,\ceil{x}}\set{x} + \sum_{s=1}^{\floor{x}} p_{k,s}  & \quad \text{if } \floor{x} \in [n_k-1] \\
		p_{k,n_k}(x-n_k) + \sum_{s=1}^{n_k} p_{k,s} & \quad \text{if } x\geq n_k  
		\end{cases} .
		\]
		Then we have $\tilde{g}_k(q) = g_k(q)$ for each $k\in [n_k]$.
		Furthermore, $\tilde{g}_k$ is continuous and essentially a linear function with $n_k-1$ points at which the slope changes. 
		Due to the ordering of the processing times the slope can only increase and hence the function is convex.
		
		Finally, note that maximal value $f_{\max}$ of the objective function can be upper bounded by $p_{\max} b^2 n$ and the maximal difference between the upper and lower bound of a variable is given by $n$.
		By plugging in the proper values, Proposition \ref{p:nfold} yields the stated running time.
	\end{proof}

	\subsection*{Scheduling with Cliques for Sum of Weighted Completion Times}
	We consider the problem of scheduling jobs non-preemptively on unrelated machines with clique incompatibility under the objective to minimize the sum of weighted completion times, i.e., $R|\conflictgraph{}|\sum w_j C_j$. 
	Recall that we are given $m$ machines forming a set $M$ and $n$ jobs forming a set $J$. 
	Each job $j \in J$ has an $m$-dimensional vector $p_j = (p_j^1, \dots, p_j^m) \in \mathbb{Z} \cup \{\infty\}$ stating that job $j$ has a processing time $p_j^i$ on machine $i \in M$. 
	Also, each job has a weight~$w_j$. 
	The jobs are partitioned into $b$ cliques. 
	Further, we introduce \textit{kinds} of jobs formally. 
	Two jobs belong to the same kind if their processing time vectors are equal and their weights are the same.
	Denote the number of \textit{job kinds} as $\vartheta$.
	We can re-write the set of jobs as $(n_1, \dots, n_{\vartheta})$ where jobs of kind $k$ appear $n_k$ times. 
	Denote by $p_{\max}$ the largest processing time and by $w_{\max}$ the largest weight occurring in the instance. 
	In the remaining of this section we prove the following theorem:
	\begin{theorem}
	The problem $R|\conflictgraph{}|\sum w_j C_j$ can be solved in FPT time parameterized by the number of machines $m$, the largest processing time $p_{\max}$ and the number of job kinds $\vartheta$.
	\end{theorem} 
	
	The main obstacle in the design of an \nfold{} IP for this setting is to formulate an appropriate objective function. 
	In \cite{DBLP:journals/scheduling/KnopK18} Knop and Kouteck{\'{y}} developed a quadratic separable convex function equivalent to the sum of completion times objective. 
	This result relies on the fact that in an optimal schedule the jobs on each machine are ordered regarding the Smith's rule, i.e., the jobs are schedules non-increasingly regarding $\rho_i(j) = w_j/p_j^i$~\cite{DBLP:journals/siamdm/GoemansW00}. 
	We may visualize this as a Gantt chart for each machine:
	Roughly speaking, it is a line of neighboring rectangles in the order of the schedule.
	The width of the $i$th rectangle is the processing time of the $i$th job on the machine and the rectangles height corresponds to the total weight of all uncompleted jobs (thus including the $i$th job).
	The area under the function, i.e. an integral of the weights of uncompleted jobs in time, corresponds to the weighted completion time and can be separated into two parts.
	One part is dependent only on the job kind and machine kind. The second one is dependent on the composition of the jobs assigned to the machine.
	By the fact that for any machine the Smith's order is optimal, the order of job kinds is known. Hence the composition is determined by the number of jobs of each kind assigned to the machine.
	Thus the second part yields a piece-wise linear convex function. 
	For details see \cite{DBLP:journals/scheduling/KnopK18}. 
	Altogether they prove:
	
	\begin{proposition}[\cite{DBLP:journals/scheduling/KnopK18}] \label{p:SepConvF}
		Let $x_1^i, \dots, x_{\vartheta}^i$ be numbers of jobs of each kind scheduled on a machine $m_i$ and let $\pi_i \colon [\{1, \dots, \vartheta\}] \rightarrow [\{1, \dots, \vartheta\}]$ be a permutation of job kinds such that $\rho_i(\pi_i(j)) \geq \rho_i(\pi_i(j+1))$ for all $1 \leq j \leq \vartheta - 1$.
Then the contribution of $m_i$ to the weighted completion time in an optimal schedule is equal to $\sum_{j=1}^{\vartheta} (1/2(z_{j}^i)^2(\rho_i(\pi_i(j)) - \rho_i(\pi_i(j+1))) + 1/2 \cdot x_{j}^i p_{j}^i w_{j})$ where $z_{j}^i = \sum_{\ell = 1}^{j}p_{\pi_i(\ell)}^i x_{\pi_i(\ell)}^i$. 
	\end{proposition}
	
	\begin{proof}
		First, let us focus on constructing the \nfold{} IP. For this result we extend the \nfold{} IP introduced in \cite{DBLP:journals/scheduling/KnopK18} and adapt the separable convex function to our needs. 
		Even though the authors separate their constraints into globally uniform and locally uniform ones, the overall number of constraints is only dependent on the parameters. 
		Thus we can shift all their constraints to the $A_i$ blocks and incooperate the clique constraints as locally uniform ones. 
		There we ensure that each machine schedules at most one job from each clique where each $B_i$ block covers one clique. 
		Let $x_{j,k}^i$ be a variable that corresponds to the number of jobs of kind $j \in \{1, \dots, \vartheta\}$ from clique $k \in \{1, \dots, b\}$ that are scheduled on machine $i \in \{1, \dots, m\}$. 
		Let $n_{j, k}$ denote number of jobs of type $j$ from clique $k$.		
		Consider the following~IP:
		\begin{align*}
		\tag{1}
		\label{11}
		& \sum_{i=1}^m x_{j,k}^i = n_{j,k} & \forall j \in \{1, \dots, \vartheta\}, \forall k \in \{1, \ldots, b\}  \\
		\tag{2}
		\label{12}
		& \sum_{k = 1}^b \sum_{\ell=1}^j x_{\pi_i(\ell), k}^i p^i_{\pi_i(\ell)} = z_j^i & \forall j \in \{1, \dots, \vartheta\}, \forall i \in \{1, \dots, m\}  \\
		\tag{3}
		\label{13}
		& \sum_{j=1}^{\vartheta} x_{j,k}^i \leq 1 & \forall i \in \{1, \dots, m\}, \forall k \in \{1, \dots, b\}  
		\end{align*}
		with lower bounds $0$ for all variables and upper bounds $x_{j,k}^i \leq 1$ and $z_j^i \leq b \cdot p_{\max}$.
		
		Let the $x_{j, k}^i$ variables form a vector $\mathbf{x}$ and the $z_j^i$ variables from a vector $\mathbf{z}$.
		Denote by $\mathbf{x}^i$ and $\mathbf{z}^i$ the corresponding subset restricted to one machine $i$. 
		The objective is to minimize the function $f(\mathbf{x}, \mathbf{z}) =\sum_{i = 1}^{m} f^i(\mathbf{x}^i, \mathbf{z}^i)=  \sum_{j=1}^{\vartheta} (1/2(z_{j}^i)^2(\rho_i(\pi_i(j)) - \rho_i(\pi_i(j+1))) + 1/2 \cdot \sum_{k=1}^b x_{j,k}^i p_{j}^i w_{j})$.  
		As we consider the altered variables $x_{j,k}^i$ over all cliques simultaneously this corresponds to the objective function from Proposition \ref{p:SepConvF}. 
		Thus, the function expresses the sum of completion times objective. 
		Further it obviously stays separable convex. 
		
		Regarding the constraint matrix, Constraint~(\ref{11}) assures that the number of jobs from a kind $j$ scheduled on the machines matches the overall number of jobs from that kind. 
		Constraint (\ref{12}) is satisfied if the $z_j^i$ variables are set as demanded in Proposition~\ref{p:SepConvF}, i.e., the jobs are scheduled with respect to the Smith's rule. 
		Finally, Constraint~(\ref{13}) assures that the number of jobs scheduled on a machine $i$ from the same clique $k$ is at most one. 
		
		We construct a schedule from the solution of the above IP in the following way:
		Place the jobs accordingly to the $x_{j,k}^i$ variables and the Smith's ratio.
		That is, assign $x_{j,k}^i$ jobs of job kind $j$ from clique $k$ to machine $i$ (note that this number is at most one due to Constraint~(\ref{13})). After assigning all jobs to a machine, place them non-increasingly regarding the Smith's ratio $\rho_i(j)$ onto the machine. 
		As we did not change the objective from \cite{DBLP:journals/scheduling/KnopK18} such a solution corresponds to an optimal one regarding the sum of weighted completion times objective.
		
		Regarding the running time we first have to estimate the \nfold IP parameters. 
		Obviously the first two constraints are globally uniform whereas the third constraint is locally uniform and repeated for each clique. 
		The parameters can be bounded by
		\begin{itemize}
			\item $n = b+1$,
			\item $t = \vartheta \cdot m$,
			\item $r = \vartheta \cdot m$,
			\item $s = m$,
			\item $\Delta = p_{\max}$,
			\item $\log(||u-\ell||_{\infty}) = \log(b \cdot p_{max})$,
			\item $\log(f_{max}) = O(\log(m \cdot b^2 \cdot p_{\max} \cdot w_{\max})) \leq O(\log(m \cdot b \cdot p_{\max} \cdot w_{\max}))$.
		\end{itemize}
		Applying Proposition \ref{p:nfold} yields a running time of
		$(p_{\max}\vartheta m)^{O(\vartheta^2m^3)} O(b \log^3(b)\log(w_{\max}))$.
		Note that the inequality constraints do not harm as we can introduce parameter many slack-variables to turn them into equality constraints. 
		Asymptotically this does not influence the running time.
	\end{proof}
	
	\subsection*{Scheduling with Parameter Many Cliques for Sum of Weighted Completion Times}
	Let us turn our attention to the same problem $R|\conflictgraph{}|\sum w_j C_j$ but parameterized by $b$, $\kappa$, $p_{\max}$ and $\vartheta$.
	Let the definitions be as in the previous section.
	The following \nfold{} IP is an extended formulation of the one from \cite{DBLP:journals/scheduling/KnopK18}.
	However, the authors did not consider cliques, thus we embed them appropriately. 
	This leads to the following theorem:
	
	\begin{theorem}
		The problem $R|\conflictgraph{}|\sum w_j C_j$ can be solved in FPT time parameterized by the number of cliques $b$, the number of machine kinds $\kappa$, the largest processing time $p_{\max}$ and the number of job kinds $\vartheta$. 
	\end{theorem} 
	
		\begin{proof}
	Regarding the variables for our IP, let $x_{j,k}^i$ denote that $x_{j,k}^i$ jobs of kind~$j \in \{1, \dots, \vartheta\}$ from clique $k \in \{1, \dots, b\}$ are scheduled on machine $i \in \{1, \dots, m\}$. 
	Further, we have $z_j^i$ for each $j \in \{1, \dots, \vartheta\}$ and $i \in \{1, \dots, m\}$. 	
	Consider the following IP:
	\begin{align*}
		\tag{1}
		\label{21}
		& \sum_{i = 1}^m  x_{j,k}^i = n_{j,k} & \forall j \in \{1, \dots, \vartheta\}, \forall k \in \{1, \dots, b\} \\
		\tag{2}
		\label{22}
		& \sum_{k = 1}^b \sum_{\ell=1}^j x_{\pi_i(\ell), k}^i p^i_{\pi_i(\ell)} = z_{j}^i & \forall j \in \{1, \dots, \vartheta\}, \forall i \in \{1, \dots, m\}  \\
		\tag{3}
		\label{23}
		& \sum_{j=1}^{\vartheta} x_{j,k}^i \leq 1 & \forall k \in \{1, \dots, b\}, \forall i \in \{1, \dots, m\}
	\end{align*}
	with lower bounds $0$ for all variables and upper bounds $x_{j,k}^i \leq 1$ and $z_{j}^i \leq b \cdot p_{\max}$.
	Again, we aim to minimize $f(\mathbf{x}, \mathbf{z}) =\sum_{i = 1}^{m} f^i(\mathbf{x}^i, \mathbf{z}^i)=  \sum_{j=1}^{\vartheta} (1/2(z_{j}^i)^2(\rho_i(\pi_i(j)) - \rho_i(\pi_i(j+1))) + 1/2 \cdot \sum_{k=1}^b x_{j,k}^i p_{j}^i w_{j})$. 
	As before, we altered the $x_j^i$ variable in the objective function by introducing more indices.
	However, as we only consider the sum of these variables this does not affect the objective and thus, by Proposition~$2$ the function maps correctly to the sum of weighted completion times objective.
	
	Regarding the IP, the constraints resemble the ones from previous IP. 
	Constraint~(\ref{21}) is satisfied if the number of jobs from kind $j$ and clique $k$ are covered by the number of jobs from that kind and clique scheduled on the machines. 
	Further, Constraint~(\ref{22}) is satisfied if the variable $z_{j}^i$ is set accordingly to Proposition $2$ i.e., the jobs are scheduled with respect to the Smith's rule. 
	The last constraint is the same as in the previous IP and it assures that the number of jobs scheduled on a machine $i$ from the same clique $k$ is at most one.
	
	A solution to the \nfold{} IP can be transformed into a schedule by placing $x_{j,k}^i$ variables of job kind $j$ and from clique $k$ onto machine $i$ (again this is at most one job due to Constraint~(\ref{23})) and ordering the jobs non-increasingly regarding the Smith's ratio $\rho_i(j)$. 
	
	To finally argue the running time, let us estimate the parameters. 
	The first constraint is globally uniform. 
	The remaining ones are locally uniform and repeated for each machine. 
	We can bound the parameters by:
	\begin{itemize}
		\item $n = m$,
		\item $t = O(\vartheta \cdot b)$, 
		\item $r = O(\vartheta \cdot b)$, 
		\item $s = O(\vartheta + b)$,
		\item $\Delta = p_{\max}$,
		\item $\log(||u-\ell||_{\infty}) = \log(b \cdot p_{max})$,
		\item $\log(f_{max}) = O(\log(m \cdot (b \cdot p_{\max}^2 + b \cdot p_{\max} \cdot w_{\max})))$.
	\end{itemize}
	
	Applying Proposition~\ref{p:nfold}  and setting up the schedule yields a running time of
	$(b\vartheta p_{\max})^{O(\vartheta^3b^3)}$ $O(m \log^2(m) \log(w_{\max}))$. 
	Again, the inequality constraints do no harm, as we can introduce few slack-variables to turn them into equality constraints. 
	Asymptotically, this does not influence the running time.
	\end{proof}
	
	\section{Open Problems}\label{sec:open_problems}
	
	While the present paper already presents quite a number of results, many interesting research directions are still open.
	For instance, the classical case of uniformly related machines $Q|\conflictgraph{}|\sum C_j$ where the processing times of the jobs depend on job sizes and machines speeds is more general than $P|\conflictgraph{}|\sum C_j$, but in turn more restricted than $R|\conflictgraph{}|\sum C_j$.
	Hence, the study of $Q|\conflictgraph{}|\sum C_j$ remains as an intriguing open problem. 
	Furthermore, we are quite interested in a more detailed study of our setting from the perspective of approximation algorithms or even FPT approximations, that is, approximation algorithms with FPT running times.
	The most obvious question in this context is probably whether a constant rate approximation for $P|\conflictgraph{},\bagmachinerestriction|\sum C_j$ is possible, given that this problem is \APXH.
	Finally, the study of further sensible classes of incompatibility graphs for the total completion time objective seems worthwhile.

	\bibliography{BagsForLipics}
	
\end{document}